\documentclass[11pt, journal,onecolumn]{IEEEtran}

\usepackage{epsfig}
\usepackage{subfigure}
\usepackage{graphicx}

\usepackage{amsmath}
\usepackage{amssymb}
\usepackage{color}
\usepackage{import}

\def\BibTeX{{\rm B\kern-.05em{\sc i\kern-.025em b}\kern-.08em
    T\kern-.1667em\lower.7ex\hbox{E}\kern-.125emX}}

\newtheorem{corollary}{Corollary}
\newtheorem{theorem}{Theorem}

\newtheorem{lemma}{Lemma}

\newtheorem{remark}{Remark}
\newtheorem{example}{Example}

\newcommand{\cX}{{\cal X}}
\newcommand{\cY}{{\cal Y}}

\newcommand{\us}{\alpha}

\newcommand{\bb}{\mathbb}

\newcommand{\blue}{\textcolor[rgb]{0,0,1}}

\newcommand{\e}{\frac{1}{2}\epsilon^2}
\newcommand{\eone}{\frac{1}{2} {\epsilon}_1^2}
\newcommand{\etwo}{\frac{1}{2} \epsilon_2^2}
\newcommand{\ezero}{\frac{1}{2}\epsilon_0^2}

\title{Euclidean Information Theory of Networks}

\author{Shao-Lun Huang, Changho Suh and Lizhong Zheng \\
\thanks{S.-L. Huang and L. Zheng are with the Research Laboratory of Electronics at Massachusetts Institute of Technology, Cambridge, USA (Email: $\mathsf{\{ shaolun, lizhong \}@mit.edu}$).

C. Suh is with the Department of Electrical Engineering at Korea Advanced Institute of Science and Technology, Daejeon, Korea (Email: $\mathsf{chsuh@kaist.ac.kr}$).

This work was supported by the National Research Foundation of Korea (NRF) Grant funded by the Korean
Government (MSIP) (No. Grant Number: 2015R1C1A1A02036561). The material in this paper was presented in part at the IEEE International Symposium on Information Theory, July 2013.
}
}

\begin{document}

\maketitle

\begin{abstract}
In this paper, we extend the information theoretic framework that was developed in earlier works to multi-hop network settings. For a given network, we construct a novel deterministic model that quantifies the ability of the network in transmitting private and common messages across users. Based on this model, we formulate a linear optimization problem that explores the throughput of a multi-layer network, thereby offering the optimal strategy as to what kind of common messages should be generated in the network to maximize the throughput. With this deterministic model, we also investigate the role of feedback for multi-layer networks, from which we identify a variety of scenarios in which feedback can improve transmission efficiency. Our results provide fundamental guidelines as to how to coordinate cooperation between users to enable efficient information exchanges across them.
\end{abstract}

\section{Introduction}

With the booming of internet and mobile communication, communication networks and social networks are rapidly growing in size and density. While the global behavior of such a large network depends on actions of individual users indeed, the sheer volume of the network makes the effect of an individual action often nonsignificant. For instance, in social networks (or stock-market networks), a public opinion (or the growth rate of wealth) is barely affected by an individual's opinion (or investment), although it is formed by their aggregation.

One natural objective for such large networks is to understand how users should design their local transmission strategies to optimize network information flow. To this end, we aim to develop an information-theoretic framework that can well model such network phenomena, as well as suggest the optimal transmission strategy of each user.

Specifically, we consider a discrete memoryless network such that the input/output distributions of each node are fixed, and each node wishes to convey information by slightly perturbing the given input distribution. In this network, we intend to investigate how a small amount of information can be efficiently conveyed to certain destinations.
Here the given distributions can be viewed as the global trend of the network, and the low-rate transmission of each node can be interpreted a nonsignificant action of an individual user. We employ mutual information in an attempt to quantify the amount of perturbation made by the users, as well as the low-rate transmission efficiency.

By employing the notion of mutual information, earlier works~\cite{BoradeZheng,AbbeZheng, HuangZheng} have made some progress towards understanding the optimal transmission strategy of users for certain networks. Specifically, Borade-Zheng~\cite{BoradeZheng} introduced a local geometric approach, based on an approximation of the Kullback-Leibler (KL) divergence, to develop a novel information-theoretic framework, and apply the framework to point-to-point channels and certain broadcast channels. Abbe-Zheng~\cite{AbbeZheng} employed the local geometric approach to address some interesting open questions in Gaussian networks. Huang-Zheng~\cite{HuangZheng} extended the framework to more general yet single-hop multi-terminal settings, and coined the \emph{linear information coupling (LIC) problems} for the associated problems (based on the framework) that will be reviewed in Section~\ref{sec:preliminaries}.

In particular~\cite{HuangZheng} developed an insightful interpretation. The key observation of~\cite{HuangZheng} is that under certain local assumptions, transmission of different types of messages, such as private and common messages, can be viewed as transmission through separated deterministic links with certain capacities. This viewpoint allows us to quantify the difficulty of broadcasting common
messages than sending private messages, as well as compute the gain of transmitting common messages. This development is particularly useful for multi-hop networks because it serves to characterize the trade-off between the gain of sending a common message and the cost that occurs in creating the common message from the previous layer.

In this work, we generalize the development into \emph{multi-hop} networks, thereby shedding some insights as to what kinds of common messages should be created in order to optimize the trade-off. Our contributions are two-fold. The first contribution is to extend the information theoretic framework in~\cite{BoradeZheng,AbbeZheng,HuangZheng} into multi-hop layered networks. Building upon this framework, we construct a deterministic network model that allows us to quantify efficiency of transmitting different types of private and common messages in the networks. This deterministic model enables us to translate the LIC problems into linear optimization problems, in which the solutions suggest what kind of common messages should be generated to optimize the throughput. With this deterministic model, we also develop an optimal local strategy for a large-scaled layered network having identical channel parameters for each layer. Specifically we demonstrate that the optimal strategy is composed of a few fundamental communication modes (to be specified in Section~\ref{sec:multi-layer}). In general, our results provide the insights of how users in a communication network should cooperate with each other to increase the efficiency of transmitting information through the network.

The second contribution of this paper is that we further generalize the framework into networks with \emph{feedback}, thereby exploring the role of feedback in multi-hop layered networks. Specifically, we consider the same layered networks but additionally include feedback links from each node to the nodes of the preceding layers. For these networks, we develop the best transmission strategy of each node that maximizes transmission efficiency.
The key technique employed here relies on our new development on network equivalence, saying that the layer-by-layer feedback strategy, which allows feedback only for the nodes in the immediately-preceding layer, yields the same performance as in the most idealistic one, where feedback is available to the entire nodes in all the preceding layers.
Moreover, we identify a variety of network scenarios in which feedback can strictly improve transmission efficiency. Our deterministic model allows us to have a deeper understanding on the nature of feedback gain: feedback offers better information routing paths, thereby making the gain of transmitting common messages effectively larger. This feedback gain is shown to be multiplicative, which is qualitatively similar to the gain in the two-user Gaussian interference channel~\cite{SuhTse}.

The rest of this paper is organized as follows. In Section~\ref{sec:preliminaries}, we review the LIC problems developed in the context of certain single-hop multi-terminal networks~\cite{BoradeZheng, AbbeZheng, HuangZheng}. The results in Section~\ref{sec:preliminaries} lead to a new type of deterministic model, which is presented in Section~\ref{sec:deterministicmodel}. In Section~\ref{sec:nofeedback}, we apply the framework to the interference channel, constructing a corresponding deterministic model. In Section~\ref{sec:MLN}, we extend this deterministic model to multi-hop layered networks, thus developing the best transmission strategy that maximizes transmission efficiency. In Section~\ref{sec:feedback}, we explore the role of feedback for multi-hop layered networks and conclude the paper with discussions in Sections~\ref{sec:extension} and~\ref{sec:conclusion}.

\section{Linear Information Coupling Problems}
\label{sec:preliminaries}

This section is dedicated to a brief review of the linear information coupling (LIC) problems which are formulated based on the local geometric approach in~\cite{BoradeZheng, AbbeZheng, HuangZheng}. Here we will summarize the local geometric approach and its application to  point-to-point channels, broadcast channels, and multiple-access channels.

In general, the LIC problems are represented in the multi-letter form. However, Huang-Zheng~\cite{HuangZheng} took the following two steps to translate them into much simpler problems: (i) translating information theory problems to linear-algebra problems, and (ii) single-letterization. In this paper, we will focus on the first step, while referring readers to~\cite{HuangZheng} for details on the single-letterization step\footnote{In general, the single-letter version is not equivalent to the corresponding multi-letter one for arbitrary networks, e.g., general $K$-user BCs. However, it is shown in~\cite{HuangZheng} that there always exist optimal finite-letter solutions. Note that our approach in this paper for solving the single-letter problems can be easily extended to their finite-letter versions, so we will consider only the single-letter problems.}.

\subsection{The Local Approximation of the Kullback-Leibler Divergence}

The key idea of the local geometric approach lies on an approximation of the Kullback-Leibler (KL) divergence ~\cite{BoradeZheng}. Let $P$ and $Q$ be two probability distributions over the same alphabet $\cal X$. We assume that $Q$ and $P$ are close to each other, i.e., $Q(x) = P(x) + \epsilon \cdot J(x)$, for some small quantity $\epsilon$. Then, using the second order Taylor expansion, the KL divergence can be written as
\begin{align}
\notag
D(P \| Q) &= \sum_{x \in {\cal X}} P(x) \ln \frac{P(x)}{Q(x)} \\ \notag
&= - \sum_{x \in {\cal X}} P(x) \ln \left( 1 + \epsilon \cdot \frac{J(x)}{P(x)} \right )  \\ \notag
&= \frac{1}{2} \epsilon^2 \cdot \sum_{x} \frac{J^2(x)}{P(x)} + o(\epsilon^2) \\ \label{eq:local_app}
&= \frac{1}{2} \epsilon^2 \cdot \| L \|^2 + o(\epsilon^2),
\end{align}
where $L = [\sqrt{P}^{-1} ] J$, and $[ \sqrt{P}^{-1} ]$ is the diagonal matrix with entries $\{ \sqrt{P(x)}^{-1}, \ x \in \cX \}$. Note that replacing $[\sqrt{P}^{-1} ]$ with $[\sqrt{Q}^{-1} ]$ in the above Euclidean norm results in only the difference of order $o(\epsilon^2)$. Hence, $D(P||Q)$ and $D(Q||P)$ are considered to be equal up to the first order approximation. From this approximation, the divergence can be viewed as the (weighted) squared Euclidean norm between two distributions. In the rest of this section, we demonstrate how this local approximation technique can be used to translate information theory problems into linear algebra problems.

\subsection{Point-to-point Channels}
\label{sec:prel_PTP}

In this section, we will first review the formulation of LIC problem in a simple context of point-to-point channels, and then explain how the local geometric approach serves to translate it into a simple linear-algebra problem. Consider a point-to-point channel with input $X \in {\cal X}$, output $Y \in {\cal Y}$, and the channel matrix $W$ associated with the channel transition probability $P_{Y|X}$.
Given some input distribution $P_X$, the LIC problem is formulated as:
\begin{equation}\label{eq:LICP_multi-letter}
\max_{U \rightarrow X \rightarrow Y :  I(U;X) \leq \frac{1}{2} \epsilon^2} I(U;Y),
\end{equation}
where $\epsilon$ is assumed to be small. The LIC problem aims at exploring the optimal transmission strategy of each node that wishes to send a small amount of message to certain destination(s) in networks. In the point-to-point setting, the following interpretation makes a connection between the above problem and the goal. Let us view $U$ as a message that the transmitter wants to send. One can then interpret $I(U;X)$ as the transmission rate of information modulated in $X$, and $I(U;Y)$ as the data rate of information that is transferred to the receiver. Unlike classical communication problems, the LIC problem targets a setting in which the amount of information is small. This is captured by the above assumption that $\epsilon$ is sufficiently small. In addition, it is assumed that\footnote{The assumption of small $I(U;X)$ does not necessarily imply $P_{X|U=u}$'s are close to $P_X$. See~\cite{Nair13, Nair14}. However, the extra assumption that $P_{X|U=u}$'s are close to $P_X$ leads to a geometric structure in the distribution spaces, which allow us to solve general network information theory problems in a systematic way. See~\cite{HuangZheng} for details. In the rest of this paper, we will employ this extra assumption and develop the geometric structure for general networks.} for all $u$ and $x$, $P_{X|U=u} (x) - P_X (x) = o(\epsilon)$. See~\cite{BoradeZheng,HuangZheng} for more detailed discussions and justifications of this formulation.

The goal of~\eqref{eq:LICP_multi-letter} is to design $P_{X|U=u}$ for different $u$, such that the marginal distribution is fixed as $P_X$, and~\eqref{eq:LICP_multi-letter} is optimized. To solve this problem, first observe that we can write the constraint as
\begin{equation}\label{eq:Px_constraint}
I(U;X) = \sum_{u} P_U (u) \cdot D(P_{X|U} (\cdot | u) \| P_X ) \leq \frac{1}{2} \epsilon^2.
\end{equation}
Thus, if we write $P_{X|U=u}$ as a local perturbation from $P_X$, i.e., $P_{X|U=u} = P_X + \epsilon \cdot J_u$, and employ the notation $L_u = [ \sqrt{P_X}^{-1} ] \cdot J_u$, then we can simplify the constraint~\eqref{eq:Px_constraint} by the local approximation~\eqref{eq:local_app} as
\begin{align*}
\sum_{u} P_U(u) \cdot \| L_u \|^2 \leq 1.
\end{align*}
Moreover, note that $U \rightarrow X \rightarrow Y$ forms a Markov relation, we have
\begin{align}
\notag
P_{Y|U=u}
= W P_{X|U=u} = WP_X + \epsilon \cdot W J_u = P_Y + \epsilon \cdot W [\sqrt{P_X}] L_u,
\end{align}
where the channel applied to the input distribution is simply viewed as the channel transition matrix $W$, of dimension $| \cY | \times | \cX |$, multiplying the input distribution as a vector.

Then, using the local approximation~\eqref{eq:local_app}, the linear information coupling problem~\eqref{eq:LICP_multi-letter} becomes a linear algebra problem:
\begin{align} \label{eq:P2P_LICP}
\max \ &\sum_{u} P_U(u) \cdot \left\| \left[\sqrt{P_Y}^{-1} \right] W \left[\sqrt{P_X} \right] \cdot L_u \right\|^2, \\ \label{eq:P2P_LICP_2}
\mbox{subject to:} \ & \sum_{u} P_U(u) \cdot \| L_u \|^2 \leq 1, \ \sum_x \sqrt{P_X(x)} L_{u} (x) = 0.
\end{align}
where the second constraint of~\eqref{eq:P2P_LICP_2} comes from
\begin{align*} 
 \sum_x \sqrt{P_X(x)} L_{u} (x) = \frac{1}{\epsilon} \sum_x (P_{X|U=u}(x) - P_X (x)) = 0.
\end{align*}
Here, we denote $B= [\sqrt{P_Y}^{-1} ] W [\sqrt{P_X} ]$ and call it the \emph{divergence transition matrix} (DTM).
Note that in both \eqref{eq:P2P_LICP} and \eqref{eq:P2P_LICP_2} the same set of weights $P_U(u)$ are used, thus the problem can be reduced to finding a direction of $L^*$, which maximizes the ratio $\| B L^*\| / \| L^* \|$, and the optimal choice of $L_u$ should be along the direction of this $L^*$ for every $u$. By linearity of the problem, scaling $L_u$ along this direction has no effect on the result. Thus, we can without loss of optimality choose $U$ as a uniformly distributed binary random variable, and further reduce the problem to:
\begin{align}
\label{eq:PTP_ApproxmiatedP}
\max_{L_{u} : \ \| L_u \|^2 \leq 1 , \ L_u \bot \sqrt{P_X}} \| B L_u \|^2,
\end{align}
where $\sqrt{P_X}$ represents a $|\cX|$-dimensional vector with entries $\sqrt{P_X(x)}$.

In order to solve~\eqref{eq:PTP_ApproxmiatedP}, we shall find $L_u$ as the right singular vector of $B$ with the largest singular value. However, the largest singular value of $B$ is $1$ with the right and left singular vectors $\sqrt{P_X}$ and $\sqrt{P_Y}$, and choosing $L_u$ as $\sqrt{P_X}$ violates the constraint $L_u \bot \sqrt{P_X}$. On the other hand, the rest right singular vectors of $B$ are orthogonal to $\sqrt{P_X}$, satisfying the constraint $L_u \bot \sqrt{P_X}$. Therefore, the optimal solution $L_u^*$ must be the right-singular vector with the \emph{second largest} singular value, and the corresponding maximum information rate is
\begin{align*}
\max || B L_u ||^2 = \sigma_{\sf smax}^2(B) =: \sigma^2.
\end{align*}
Here $\sigma_{\sf smax}(B)$ denotes the second largest singular value of $B$, which we define as $\sigma$. This shows that the problem is reduced to a simple linear-algebra problem of finding the fundamental direction $L_u^*$ that maximizes the amount of information $I(U;Y)$ that flows into the receiver.

\begin{example} Consider a quaternary-input binary-output point-to-point channel:
\begin{align*}
Y = \left\{
      \begin{array}{ll}
        X \oplus Z_1, & \hbox{$X \in \{ 0,1 \}$;} \\
        (X \mod 2) \oplus Z_2, & \hbox{$X \in \{2,3 \}$,}
      \end{array}
    \right.
\end{align*}
where $Z_1 \sim \sf Bern(\frac{1}{2})$ and $Z_2 \sim \sf Bern(\alpha)$. The probability transition matrix is then computed as
\begin{align*}
W = \left[
      \begin{array}{cccc}
        \frac{1}{2} & \frac{1}{2} & 1-\alpha & \alpha \\
        \frac{1}{2} & \frac{1}{2} &\alpha   & 1-\alpha \\
      \end{array}
    \right].
\end{align*}
Suppose that $P_X$ is fixed as $[\frac{1}{4}, \frac{1}{4}, \frac{1}{4}, \frac{1}{4}]^T$.
We can then compute $P_Y = W P_X = [\frac{1}{2}, \frac{1}{2}]^T$ and $B=\frac{\sqrt{2}}{2}W$. A simple computation gives:
\begin{align*}
L_u^* = \frac{1}{\sqrt{2}} [ 0, 0, 1, -1]^T, \; \sigma^2 = ||B L_u^* ||^2 = \frac{1}{2} (1- 2\alpha)^2.
\end{align*}
This solution is intuitive. Note that when $X \in \{0,1 \}$, it passes through a zero-capacity channel with $Z_1 \sim {\sf Bern} (\frac{1}{2})$. On the other hand, when $X \in \{ 2,3 \}$, the channel is a binary symmetric channel with $\alpha$. Therefore, information can be transferred only when $X \in \{2,3\}$, which matches the solution of $L_u^*$ as above. Note that $L_u^*$ contains non-zero elements only for the third and fourth entries corresponding to $X=2$ and $X=3$ respectively. When $\alpha \approx \frac{1}{2}$, the channel w.r.t $X \in \{2,3\}$ is very noisy. As $\alpha$ is far away from $\frac{1}{2}$, however, the channel is less noisy, thus delivering more information. This is reflected in the form of $\sigma^2$ as above. $\square$
\end{example}

\subsection{Broadcast Channels}
\label{sec:prel_BC}

Now, let us consider the LIC problem of broadcast channels. Suppose that a two-receiver discrete memoryless broadcast channel with input $X \in {\cal X}$ and two outputs $(Y_1,Y_2) \in {\cal Y}_1 \times {\cal Y}_2$, is specified by the memoryless channel matrices $W_1$ and $W_2$. These channel matrices specify the conditional distributions of the output signals at two receivers as $W_k ( y_k | x ) = P_{Y_k|X} ( y_k | x )$ for $k = 1,2$.
Let $U_0$ be a common message intended for both receivers; and $U_1$, $U_2$ be private messages intended for receivers 1 and 2 respectively. Assume that $(U_0,U_1,U_2)$ are mutually independent and $P_X$ is fixed. Let $(R_1, R_2, R_0)$ be the corresponding information rates.

For this setting, the LIC problem is formulated as the one that maximizes a rate region ${\cal R}_{\sf BC}$ such that
\begin{align} \label{eq:BC_capacity_region}
\begin{split}
R_1 &\leq  I(U_1; Y_1), \; R_2 \leq I(U_2; Y_2), \\
R_0 &\leq  \min \{ I(U_0;Y_1), I(U_0; Y_2) \},
\end{split}
\end{align}
under the locality assumption of
\begin{align*}
\begin{split}
&I(U_1; X) \leq \eone,\; I(U_2; X) \leq \etwo, \\
&I(U_0;X) \leq \ezero, \; \epsilon_1^2 + \epsilon_2^2 + \epsilon_0^2  = \epsilon^2.
\end{split}
\end{align*}
Here $( U_0 , U_1 , U_2 ) \rightarrow X \rightarrow (Y_1, Y_2)$ forms a Markov relation and $\epsilon$ is assumed to be some small quantity.

While a natural extension of the point-to-point-channel locality assumption is $I(U_1,U_2,U_0; X) \leq \e$, it can be shown that~\cite{HuangZheng} the resultant rate region ${\cal R}_{\sf BC}$ with this assumption is the same as considering the above three separate assumptions instead. Note that $I(U_1,U_2,U_0; X) \leq \e$ captures the tradeoff between $(R_1,R_2,R_0)$ in an aggregated manner, thus making the optimization involved. On the other hand, under the separate assumptions, the tradeoff is captured only by $\epsilon_1^2 + \epsilon_2^2 + \epsilon_0^2  =: \epsilon^2 \ll 1$: given that $\epsilon^2$ is appropriately allocated to $(\epsilon_1^2, \epsilon_2^2, \epsilon_0^2)$, there is no tension between those rates. Hence, this simplification enables us to reduce the problem to three independent sub-problems: two are w.r.t. private messages $(U_1, U_2)$, and the last is w.r.t. the common message $U_0$.

The optimization problems w.r.t. the private messages are the same as in the point-to-point channel case: for $k=1,2,$
\begin{align*}
\max I(U_k;Y_k) = \frac{1}{2} \epsilon_k^2 \cdot \sigma_k^2 + o (\epsilon^2),
\end{align*}
where $\sigma_k = \sigma_{\sf smax} (B_k)$, and $B_k =[\sqrt{P_{Y_k}}^{-1}] W_k [\sqrt{P_{X}}]$.
Thus, the main focus here is the optimization of the common information rate. Suppose that $P_{X|U_0=u_0} = P_X + \epsilon \cdot J_{u_0} $, and $L_{u_0} = [\sqrt{P_{X}}^{-1} ]J_{u_0}$. Using similar arguments, we can then reduce the problem to:
\begin{align} \label{eq:BC_common}
\max_{L_{u_0} : \ \| L_{u_0} \|^2 \leq 1 , \ L_{u_0} \bot \sqrt{P_X}} \min \left\{  \| B_1 L_{u_0} \|^2,
\|  B_2 L_{u_0} \|^2 \right\}.
\end{align}
Now, this problem is simply a finite dimensional convex optimization problem, which can be easily solved. Let $\sigma_0^2$ be the maximum value w.r.t. the $L_{u_0}^*$.

\begin{example}
\label{example:BC}
Consider a quaternary-input binary-outputs BC: for $k \in \{ 1,2 \}$,
\begin{align*}
Y_k = \left\{
      \begin{array}{ll}
        X \oplus Z_{k1}, & \hbox{$X \in \{ 0,1 \}$;} \\
        (X \mod 2) \oplus Z_{k2}, & \hbox{$X \in \{2,3 \}$,}
      \end{array}
    \right.
\end{align*}
where $Z_{11}, Z_{22} \sim \sf Bern(\frac{1}{2})$ and $Z_{12}, Z_{21} \sim \sf Bern(\alpha)$. The transition probability matrices are computed as \begin{align*}
&W_1 = \left[
      \begin{array}{cccc}
        \frac{1}{2} & \frac{1}{2} &1-\alpha & \alpha \\
        \frac{1}{2} & \frac{1}{2} &\alpha   & 1-\alpha \\
      \end{array}
    \right], \\
&W_2 = \left[
      \begin{array}{cccc}
        1-\alpha & \alpha & \frac{1}{2} & \frac{1}{2}  \\
        \alpha & 1 - \alpha & \frac{1}{2} &  \frac{1}{2} \\
      \end{array}
    \right].
\end{align*}
Suppose that $P_X$ is fixed as $[\frac{1}{4}, \frac{1}{4}, \frac{1}{4}, \frac{1}{4}]^T$. We can then get $P_{Y_1} = P_{Y_2 } = [\frac{1}{2}, \frac{1}{2}]^T$. This allows us to compute $B_k = \frac{\sqrt{2}}{2} W_k, (k=1,2)$. With a simple linear-algebra calculation, we obtain
\begin{align*}
&L_{u_1}^* = \frac{1}{\sqrt{2}} [0, 0, 1, -1]^T, \;
\sigma_1^2 = \frac{1}{2} (1- 2\alpha)^2; \\
&L_{u_2}^* = \frac{1}{\sqrt{2}} [1, -1, 0, 0]^T, \;
\sigma_2^2 = \frac{1}{2} (1- 2\alpha)^2; \\
&L_{u_0}^* = \frac{1}{2} [1, -1, -1, 1]^T,
\; \sigma_0^2 = \frac{1}{4} \left( 1- 2 \alpha \right)^2.
\end{align*}
Here, one can see the difficulty of delivering common message, as compared to private message transmission. Note that $\sigma_0^2$ is half of the $\sigma_1^2 (= \sigma_2^2)$. This example represents an extreme case where $\sigma_0^2$ is minimized for all possible channels having the same $\sigma_1$ and $\sigma_2$, and thus the gap between $\sigma_0$ and $\sigma_1 (=\sigma_2)$ is maximized. Note that $\sigma_0^2$ has a trivial lower bound. It must be greater than a naive transmission rate:
$\min \{\lambda \sigma_1^2, (1-\lambda) \sigma_2^2 \}$, which can be achieved by privately sending a message first to receiver 1 with the fraction $\lambda$ of time and later to receiver 2 with the remaining fraction $(1-\lambda)$ of time.
This naive rate can be maximized as:
\begin{align}
\label{eq:naivemaxrate}
\max_{0 \leq \lambda \leq 1} \min \{\lambda \sigma_1^2, (1-\lambda) \sigma_2^2 \} = \frac{ \sigma_1^2 \sigma_2^2}{\sigma_1^2 + \sigma_2^2 }.
\end{align}
In this example, this rate is maximized as $\frac{\sigma_1^2}{2}$, which coincides with $\sigma_0^2$. $\square$
\end{example}

\subsection{Multiple-access Channels}
\label{sec:prel_MAC}

Now, let us consider the LIC problem of multiple-access channels. Suppose that the multiple-access channel has two inputs $X_1 \in \cX_1$, $X_2 \in \cX_2$, and one output $Y \in \cY$. The memoryless channel is specified by the channel matrix $W$, where $W (y | x_1, x_2 ) = P_{Y|X_1,X_2} (y|x_1,x_2)$ is the conditional distribution of the output signals.
We want to communicate three messages $(U_1, U_2, U_0)$ to the receiver with rates $(R_1,R_2, R_0)$, where $U_1$ and $U_2$ are privately known by transmitter $1$ and $2$ respectively, and $U_0$ is the common source known to both transmitters. Then, the LIC problem for the MAC is formulated as the one that maximizes a rate region ${\cal R}_{\sf MAC}$:
\begin{align} \label{eq:MAC_capacity_region}
R_0 \leq  I(U_0;Y), \ R_1 \leq  I(U_1;Y), \ R_2 \leq  I(U_2;Y),
\end{align}
such that $U_0 \rightarrow (X_1,X_2) \rightarrow Y$, $U_1 \rightarrow X_1 \rightarrow Y$, $U_2 \rightarrow X_2 \rightarrow Y$, and the local constraints are:
\begin{align*}
\begin{split}
&I(U_1; X_1) \leq \eone, \ I(U_2; X_2) \leq \etwo,\\
&I(U_0;X_1,X_2) \leq \ezero,  \ \epsilon_1^2 + \epsilon_2^2 + \epsilon_0^2  = \epsilon^2.
\end{split}
\end{align*}
Again, $\epsilon$ is assumed to be some small quantity.

Define the DTMs $B_k =[\sqrt{P_{Y}}^{-1}] W_k [\sqrt{P_{X_k}}]$, for $k= 1,2$, where
$$W_k ( y | x_k ) = \sum_{x_{3-k} \in \cX_{3-k}} W (y|x_1, x_2) P_{X_{3-k}} (x_{3-k}).$$
Two optimization problems w.r.t. private messages are the same as in the point-to-point channel case: $\max I(U_k;Y) = \frac{1}{2} \epsilon_k^2 \sigma_k^2 + o (\epsilon^2)$ where  $\sigma_k = \sigma_{\sf smax}(B_k)$.

Now suppose that
\begin{align*}
 P_{X_i|U_0 = u_{0}}  = P_{X_i}  + \epsilon_{0} \cdot J_{i,u_0}.
\end{align*}
Since $X_1$ and $X_2$ are conditionally independent given $U_0$, we can write $P_{X_1 X_2|U_0 = u_0}$ as
\begin{align}
\notag
P_{X_1, X_2 | U_0 = u_0}
= P_{X_1 | U_0 = u_0} \otimes P_{X_2 | U_0 = u_0} = P_{X_1} \otimes P_{X_2} +  \epsilon_{0} \cdot  J_{1,u_0} \otimes P_{X_2} + \epsilon_{0} \cdot P_{X_1} \otimes J_{2,u_0} + O(\epsilon^2). \\ \notag
\end{align}
Then, the condition $I(U_0;X_1,X_2) \leq \ezero$ can be written as
$$\sum_{u_0} P_{U_0} (u_0) \cdot \left\| L_{u_0} \right\|^2 \leq 1,$$
where $L_{u_0} = \left[ [ \sqrt{P_{X_1}}^{-1} ] J_{1,u_0}^T \ \  [ \sqrt{P_{X_2}}^{-1} ] J_{2,u_0}^T \right]^T$. Moreover, we can write $P_{Y|U_0 = u_0}$ as
\begin{align*}
P_{Y | U_0 = u_0}
= W \cdot P_{X_1, X_2 | U_0 = u_0}
= P_{Y} + \epsilon_{0} W_1 J_{1,u_0} + \epsilon_{0} W_2 J_{2,u_0} + O(\epsilon^2)
\end{align*}
so $I(U_0;Y)$ can be written as
$$\sum_{u_0} P_{U_0} (u_0) \cdot \left\| B_0 L_{u_0} \right\|^2,$$
where $B_0 = [B_1 \ B_2]$. Therefore, the optimization problem w.r.t. the common message can be reduced to
\begin{align} \label{eq:MAC_common}
\max_{L_{u_0} : \| L_{u_0} \|^2 \leq 1}  \left\| B_0 L_{u_0} \right\|^2.
\end{align}
Observe that unlike the point-to-point channel case, the $L_{u_0}$ has to respect the constraint that the first $|\cX_1|$ entries of $L_{u_0}$ (an $|\cX_1|$-dimensional vector) is orthogonal to $\sqrt{P_{X_1}}$, and the last $|\cX_2|$ entries of $L_{u_0}$ is orthogonal to $\sqrt{P_{X_2}}$. Nevertheless it is shown in~\cite{HuangZheng} that the optimal $L_{u_0}$ in~\eqref{eq:MAC_common} is still the right singular vector of $B_0$ with the second largest singular value. Hence, the maximum of~\eqref{eq:MAC_common} is $\frac{1}{2} \epsilon_0^2 \sigma_0^2$ where $\sigma_0^2 =  \sigma_{\sf smax}^2 ([B_1 \; B_2])$.

\begin{example}
\label{example:MAC}
Consider a quaternary-inputs binary-output MAC with
\begin{align*}
&P(0|x_1x_2) = \left\{
                 \begin{array}{ll}
                   \frac{1}{3} ( 2- \alpha ), & \hbox{$x_1x_2=(00,01,02,10, 11,12)$;} \\
                   \alpha, & \hbox{$x_1x_2=(03, 13,23, 33)$;} \\
                   \frac{1}{3} (4 - 5 \alpha), & \hbox{$x_1x_2=(20,21,22,32)$;} \\
                   \frac{1}{3} ( -2 + 7 \alpha), & \hbox{$x_1x_2=(30,31)$,}
                 \end{array}
               \right. \\
&P(1|x_1x_2) = 1 - P(0|x_1x_2), \;\forall (x_1,x_2).
\end{align*}
Here we assume that $\frac{2}{7} \leq \alpha \leq \frac{5}{7}$, which allows us to have a valid probability distribution. Suppose that both $P_{X_1}$ and $P_{X_2}$ are fixed as $[\frac{1}{4}, \frac{1}{4}, \frac{1}{4}, \frac{1}{4}]^T$. The probability transition matrices are then given by
\begin{align*}
W_1 = W_2 = \left[
      \begin{array}{cccc}
        \frac{1}{2} & \frac{1}{2} & 1-\alpha & \alpha \\
        \frac{1}{2} & \frac{1}{2} & \alpha   & 1-\alpha \\
      \end{array}
    \right].
\end{align*}
We can then compute $B_1 = B_2 = \frac{\sqrt{2}}{2} W_1$.
Hence, we get the same $(L_{u_k}^*, \sigma_k^2)$ as that in Example~\ref{example:BC} for $k=1,2$. For $(L_{u_{0}}^{*},\sigma_{0}^2)$, we obtain
\begin{align*}
L_{u_{0}}^{*} = \frac{1}{2} [0,0,1,-1 | 0,0,1,-1]^T , \quad \sigma_0^2 = (1-2\alpha)^2.
\end{align*}

Here we can see a gain due to coherent combining of the transmitted signals. Notice that the common rate $\sigma_0^2$ is double the private rate $\sigma_1^2 = \sigma_2^2$. One can interpret this as a so-called beamforming gain that is widely used to indicate the coherent combining gain in the context of multi-antenna Gaussian channels. $\square$
\end{example}

\section{A New Deterministic Model}
\label{sec:deterministicmodel}

The local geometric framework in Section~\ref{sec:preliminaries} provides a systematic approach in exploring the LIC problems. It turns out that this approach allows us to abstract arbitrary communication networks with a few key parameters induced by the networks, thus developing a novel deterministic model. In this section, we construct deterministic models for the point-to-point, broadcast and multiple-access channels discussed in the preceding section, and will extend to more general communication networks in the following sections.

Prior to describing our model, we emphasize three distinguishing features of the model with a comparison to one popular deterministic model: the Avestimehr-Diggavi-Tse (ADT) model~\cite{Salman:IT11}.
\begin{itemize}
\item \emph{Target channels:} While the ADT model is intended for capturing key properties of wireless Gaussian channels, our model aims at arbitrary discrete-memoryless channels.
\item \emph{Approximation:} In the ADT model, approximation to Gaussian channels is accurate when links have high signal-to-noise ratios. On the other hand, our model relies upon the Euclidean approximation and hence it is accurate as long as the channels are assumed to be very noisy, i.e., $P_{X|U=u}$ being close to $P_X$. The locality assumption puts limitations to our model in approximating general not-very-noisy channels.
\item \emph{Signal interactions in the noise-limited regime:} The ADT model focuses on the interaction of transmitted signals rather than on background noises, thus well representing the interference-limited regime, where the noise power is negligible compared to signal powers. Our model, however, can well represent noise-limited regimes in which a beamforming gain often occurs. Moreover, even for very noisy channels, signal interactions can be captured in our model. This is a significant distinction with respect to the ADT model targeted for Gaussian channels. Note that for very noisy Gaussian channels, signal interactions are completely ignored as the channels are considered as multiple point-to-point links in the noise-limited regime.
\end{itemize}

\begin{remark}
\label{remark:ADTvsLIC}
While our model does not well approximate not-very-noisy channels which often represent many realistic communication scenarios,
it still plays a role in some realistic networks. One such example is a cognitive radio network in which secondary users wish to exchange their information while minimizing interference to the existing communication network for primary users.
By modeling the encoding of the secondary users' signals as superposition coding to existing primary signals, we can formulate an LIC problem that intends to characterize the tradeoff between the communication rates of the secondary users and the interference to the existing communication network.
In Section~\ref{sec:LGA}, we will provide more detailed discussions on this, and also show the potential of our model to a wide range of other interesting applications beyond communications.
\end{remark}


\emph{Notations}: For illustrative purpose, we shall use the following notations for the rest of this paper. Let $\delta$ and $\delta_k$ be $\frac{1}{2} \epsilon^2$ and $\frac{1}{2} \epsilon_k^2$ respectively. In fact, we assume that $\delta$ is a small value, as it allows us to exploit the local approximation to derive capacity regions. However, once the capacity regions are obtained, the $\delta$ acts only as a scaling factor. So for simplicity, we normalize the regions by replacing $\delta$ with 1.
In addition, in order to distinguish the local-approximation-based capacity region from the traditional one, we shall call it the \emph{linear information coupling (LIC) capacity region}. With slight abuse of notations, we will use the notation ${\cal C}$ (usually employed to indicate the conventional capacity region) to denote the LIC capacity region. We will also use the notation $C_{\sf sum}$ to indicate the LIC sum capacity.

\subsection{Point-to-point Channels}

For a point-to-point channel, from Section II-A, the LIC capacity  is simply $I(U;X) \approx \delta \cdot \sigma^2$. This naturally leads us to model the point-to-point channel as a single bit-pipe with capacity $\sigma^2$. Here the quantity $\sigma^2$ can be computed simply as the second largest singular value of the DTM. Importantly, note that this deterministic model provides a general framework as it can abstract \emph{every} discrete-memoryless point-to-point channel with a single quantity $\sigma^2$.

\subsection{Broadcast Channels}

For a general broadcast channel, the LIC capacity region~\eqref{eq:BC_capacity_region} is derived as
\begin{align*}
{\cal C}_{\sf BC} = \bigcup_{\delta_1 + \delta_2 + \delta_0 \leq 1} \left \{ (R_1,R_2,R_0): R_k \leq \delta_k \sigma_k^2, k \in [0:2] \right \},
\end{align*}
where $\sigma_k$'s can be computed as in Section~\ref{sec:prel_BC}.
This simple formula of the region leads us to model a broadcast channel as three bit-pipes, each having capacity $\delta_k \sigma_i^2$. Unlike traditional wired networks, the capacities of these bit-pipes are flexible: $\delta_k \sigma_k^2$ can vary depending on different allocations of $(\delta_1,\delta_2,\delta_0)$ subject to $\delta_1 + \delta_2 + \delta_0 \leq 1$. Hence, the LIC capacity region is of the shape as shown in the right figure of Fig.~\ref{fig:BC}.

\begin{figure}[t]
\begin{center}
{\epsfig{figure=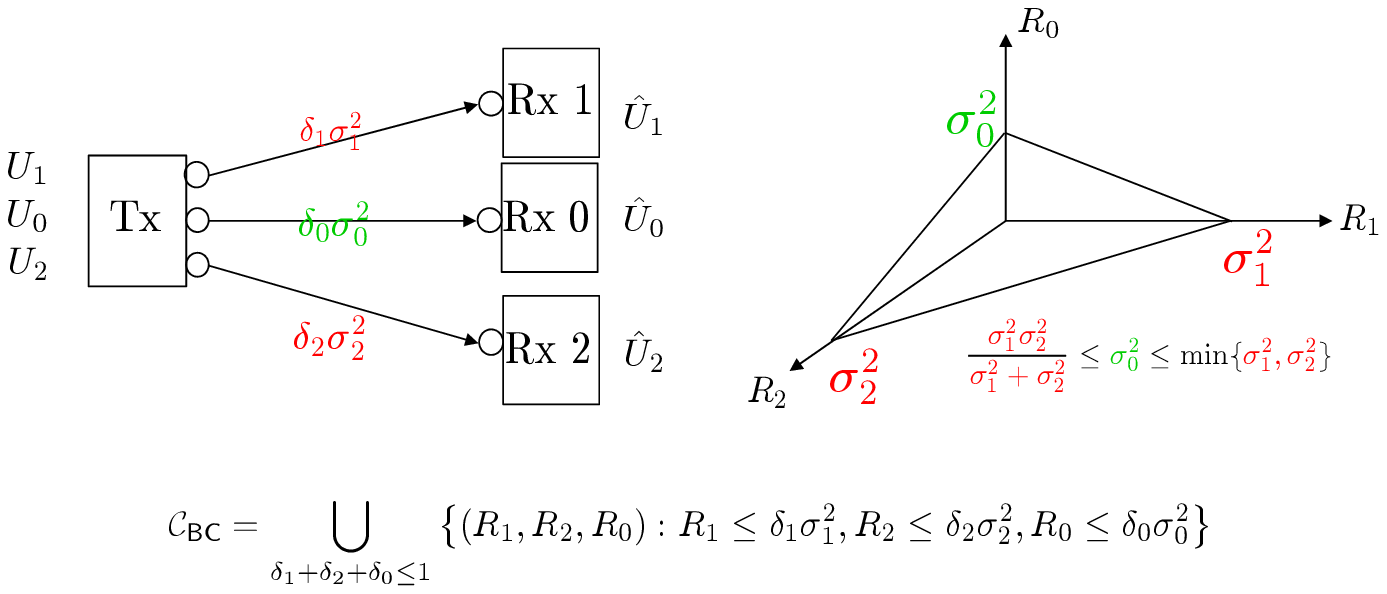, angle=0, width=0.7\textwidth}}
\end{center}
\caption{The bit-pipe deterministic model for discrete-memoryless broadcast channels. The LIC capacity region leads us to abstract a BC as a deterministic channel with three bit-pipes, each having the capacity of $\delta_k \sigma_k^2$. Note that the capacity $\delta_k \sigma_k^2$ can change depending on an allocation of $(\delta_1,\delta_2,\delta_0)$. Here we normalize the capacity region by $\delta$. Rx $k$ indicates a virtual terminal that decodes only $U_k$, for $k=0,1,2$. Hence, physical-Rx $k$ consists of virtual-Rx $k$ and virtual-Rx $0$, for $k=1,2$.
} \label{fig:BC}
\end{figure}

The left figure in Fig.~\ref{fig:BC} shows a pictorial representation of our deterministic model for discrete-memoryless broadcast channels. Here physical-Rx $k$ wishes to decode its private message $U_k$ as well as the common message $U_0$.
So we can represent physical-Rx $k$ by two \emph{virtual} receivers, say Rx $k$ and Rx $0$, which intend to decode $U_{k}$ and $U_0$ respectively. Employing the virtual receivers, we now model the broadcast channel with one transmitter and three receivers in which each receiver decodes its individual message.
Here the circles indicate bit-pipes intended for transmission of different messages. For instance, the top circle indicates a bit-pipe w.r.t. the $U_1$-message transmission. Note that different types of messages are delivered via parallel channels, identified by circles.

Another significant distinction w.r.t. the traditional wired network model is that channel parameters $(\sigma_1^2, \sigma_2^2, \sigma_0^2)$ have to respect the inequality that intrinsically comes from the structure of the broadcast channel:
\begin{align}
\label{eq:BClambda0}
\frac{ \sigma_1^2 \sigma_2^2}{\sigma_1^2 + \sigma_2^2 } \leq \sigma_0^2 \leq \min \{ \sigma_1^2, \sigma_2^2 \}.
\end{align}
Notice that the lower bound can be achieved as shown in Example~\ref{example:BC}.
This equality corresponds to the case, where the two optimal perturbation vectors for each of the two users are somehow orthogonal, and it is difficult to find a communication scheme that conveys much information to both receivers simultaneously. On the other hand, the equality of the upper bound holds  when the two optimal communication directions of two users are aligned with each other, so that one can design a perturbation vector that broadcasts information to both receivers efficiently. Moreover, the upper bound implies that common-message transmission requires more communication resources than private-message transmission does. Following the procedure in Section~\ref{sec:prel_BC}, one can explicitly computing $\sigma_k$'s, thus quantifying the cost difference between common-message and private-message transmissions.

In addition, in this deterministic model, the trade-off between $(R_1,R_2,R_0)$ can be well adjusted with $(\delta_1,\delta_2,\delta_0)$ subject to $\delta_1 + \delta_2 + \delta_0 \leq 1$. This trade-off can be precisely evaluated from $\mu$-sum-rate maximization, which can be carried out via a simple LP problem formulation as follows:
\begin{align*}
\max \sum_{k=0}^{2} \mu_k \cdot (\delta_k \sigma_k^2): \;\; \textrm{s.t.} \sum_{k=0}^{2} \delta_k \leq 1.
\end{align*}
In the case of the sum-rate maximization where $\mu_k=1, \forall k$, we can get {$C_{\sf sum} = \max \{ \sigma_1^2,\sigma_2^2, \sigma_0^2 \} = \max \{ \sigma_1^2,\sigma_2^2 \}$. Here we have used~\eqref{eq:BClambda0}. This solution implies that common-message transmission is more expensive, and hence choosing a more capable link among private-message bit-pipes yields the maximum sum rate.

\subsection{Multiple-Access Channels}

\begin{figure}[t]
\begin{center}
{\epsfig{figure=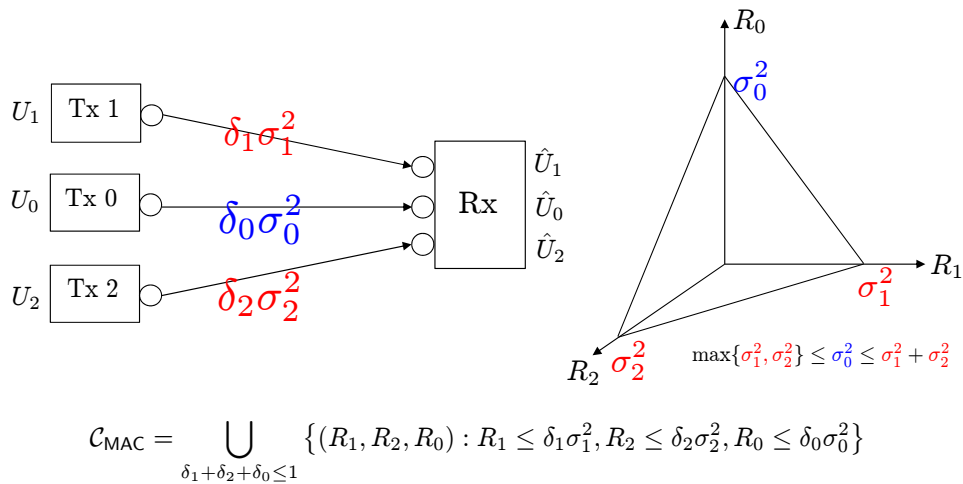, angle=0, width=0.7\textwidth}}
\end{center}
\caption{The bit-pipe deterministic model for multiple-access channels. A discrete-memoryless MAC can be modeled as three bit-pipes where the capacity of each bit-pipe is $\delta_k \sigma_k^2$. Unlike BCs, virtual transmitters are employed. Tx $k$ indicates a virtual terminal that sends only $U_k$, for $k=0,1,2$. Hence, physical-Tx $k$ consists of virtual-Tx $k$ and virtual-Tx $0$, for $k=1,2$.
} \label{fig:MAC}
\end{figure}

The LIC capacity region~\eqref{eq:MAC_capacity_region} for the multiple-access channel is derived as
\begin{align*}
{\cal {C}}_{\sf MAC} = \bigcup_{\delta_1 + \delta_2 + \delta_0 \leq 1} \left \{ (R_1,R_2,R_0): R_k \leq \delta_k \sigma_k^2, k \in [0,2] \right \},
\end{align*}
where $\sigma_k$'s can be computed as in Section~\ref{sec:prel_MAC}. Therefore, any discrete-memoryless MAC can be modeled as three bit-pipes, each having capacity $\delta_k \sigma_k^2$. See Fig.~\ref{fig:MAC}. Applying similar ideas as in the broadcast channel, we model physical-Tx $k$ by two virtual transmitters, say Tx $k$ and Tx $0$, which wishes to send the private message $U_{k}$ and the common message $U_{0}$ respectively. So the multiple access channel is modeled with three transmitters and one receiver.

Similarly, channel parameters $(\sigma_1^2, \sigma_2^2, \sigma_0^2)$ here should also satisfy the inequality that comes intrinsically from the MAC structure:
\begin{align}
\label{eq:MACsigma}
\max \{ \sigma_1^2, \sigma_2^2 \}  \leq \sigma_0^2 \leq \sigma_1^2 + \sigma_2^2.
\end{align}
The lower bound of~\eqref{eq:MACsigma} is straightforward. To see the upper bound,
notice that for any valid perturbation vector $L = [L_1^T \ L_2^T]^T$,
\begin{align*}
\left \| B_0 L
\right \|^2
\leq \left( \left \| B_1 L_{1} \right \| +  \left \| B_2 L_{2} \right \| \right)^2
\leq (\sigma_1 \| L_1 \| + \sigma_2 \| L_2 \|)^2
 \leq \sigma_1^2 + \sigma_2^2.
\end{align*}
Here the first inequality is due to the triangle inequality. The second inequality follows from the definition of $\sigma_1$ and $\sigma_2$: $\sigma_k$ denotes the second largest singular value of $B_k$, $k=1,2$. The third inequality comes from the Cauchy-Schwarz inequality and the unit-norm constraint: $ || L ||^2  = \| L_1 \|^2 +  \| L_2 \|^2 \leq 1$. Importantly, note that both transmitters share the knowledge of the common message, and hence they can cooperate each other in sending the common message efficiently. This is reflected in the upper bound of~\eqref{eq:MACsigma}, being interpreted as the coherent combining gain (or the beamforming gain).

Moreover, the trade-off between $(R_1,R_2,R_0)$ can be evaluated from $\mu$-sum-rate maximization. For example, the LIC sum capacity is given by $C_{\sf sum} = \max \{ \sigma_1^2,\sigma_2^2, \sigma_0^2 \} = \sigma_0^2$, obtained via maximizing the coherent combining gain.

Unlike the ADT model, our model can capture signal interactions even for non-negligible noisy channels. This is demonstrated through the following example.

\begin{example}
Consider a binary-inputs binary-output MAC with
\begin{align*}
&P(0|x_1x_2) = \left\{
                 \begin{array}{ll}
                   1-\alpha, & \hbox{$x_1x_2=(00,11)$;} \\
                   \alpha, & \hbox{$x_1x_2=(01, 10)$.}
                 \end{array}
               \right. \\
&P(1|x_1x_2) = 1 - P(0|x_1x_2), \;\forall (x_1,x_2).
\end{align*}
In fact, this is a binary addition channel:
\begin{align*}
Y = X_1 \oplus X_2 \oplus Z,
\end{align*}
where $Z \sim {\sf Bern} (\alpha)$. Suppose that both $P_{X_1}$ and $P_{X_2}$ are fixed as $[\frac{1}{2}, \frac{1}{2}]^T$. The probability transition matrices are then given by
\begin{align*}
W_1 = W_2 = \left[
      \begin{array}{cc}
        \frac{1}{2} & \frac{1}{2} \\
        \frac{1}{2} & \frac{1}{2} \\
      \end{array}
    \right].
\end{align*}
We can then compute $B_1 = B_2 = W_1$, thus yielding $\sigma_1^2=\sigma_2^2=\sigma_0^2=0$.

We now consider a different MAC where the above joint probability distribution is slightly changed as follows:
\begin{align*}
&P(0|x_1x_2) = \left\{
                 \begin{array}{ll}
                   1-\alpha, & \hbox{$x_1x_2=(00,{10})$;} \\
                   \alpha, & \hbox{$x_1x_2=(01, {11})$.}
                 \end{array}
               \right. \\
&P(1|x_1x_2) = 1 - P(0|x_1x_2), \;\forall (x_1,x_2).
\end{align*}
The only difference here is that the probabilities $P(y|10)$ and $P(y|11)$ are simply swapped each other. This simple change yields different values of $(\sigma_1^2, \sigma_2^2, \sigma_0^2)$. Note that in this case,
\begin{align*}
W_1  = \left[
      \begin{array}{cc}
        \frac{1}{2} & \frac{1}{2} \\
        \frac{1}{2} & \frac{1}{2} \\
      \end{array}
    \right], \; W_2 = \left[
      \begin{array}{cc}
        1-\alpha & \alpha \\
        \alpha &  1-\alpha\\
      \end{array}
    \right],
\end{align*}
thus yielding $(\sigma_1^2,\sigma_2^2,\sigma_0^2)=(0,\frac{(1-2\alpha)^2}{2},\frac{(1-2\alpha)^2}{2})$.
Therefore, we can see that even for non-negligible noisy channels, signal interactions are well captured in our model. $\square$
\end{example}

We now generalize this deterministic model to arbitrary discrete-memoryless networks.
Specifically we will first construct a deterministic model for  interference channels in Section~\ref{sec:nofeedback}, and then extend to more general networks in the following sections.

\section{Interference Channels}
\label{sec:nofeedback}

The quantifications of the channel parameters in~\eqref{eq:BClambda0} and~\eqref{eq:MACsigma} in Section~\ref{sec:deterministicmodel} shed significant insights into exploring transmission efficiency in more general networks. Specifically~\eqref{eq:MACsigma} suggests that common-message transmission in the MAC is more advantageous due to the coherent combining gain. This motivates us to create common messages as much as possible. On the other hand,~\eqref{eq:BClambda0} suggests that it consumes more network resources to generate such common messages than the private-message generation. Hence, there is a fundamental trade-off between the cost of generating common messages and the benefit from transmitting common messages. With the framework established in the previous sections, we now intend to investigate the trade-off relation, thereby optimizing communication rates of networks. To this end, we will first explore interference channels in this section.

For an interference channel with two transmitters and two receivers, there are 9 types of messages $U_{ij}$ where $i,j = 0,1,2$. Here $U_{ij}$ indicates a message from virtual-Tx $i$ to virtual-Rx $j$, $i, j \in [0:2]$. Note that $U_{i0}$ denote a common message (w.r.t. virtual-Tx $i$) intended for both receivers, while $U_{0j}$ indicates a common message (w.r.t. virtual-Rx $j$) accessible by both transmitters. Then, the LIC problem for the interference channel is the one that maximizes a rate region such that
\begin{align} \label{eq:IC_locality_assumption}
&R_{ij} \leq I(U_{ij}; Y_{j}), \ \forall i, \ j \neq 0, \\
&R_{i0} \leq \min \left\{ I(U_{i0}; Y_{1}), I(U_{i0}; Y_{2}) \right\} \  \forall i,
\end{align}
subject to the constraints:
\begin{align*}
\begin{split}
&I(U_{ij}; X_i) \leq \delta_{ij}, \ i \neq 0, \ \forall j, \\
&I(U_{0j}; X_1,X_2) \leq \delta_{0j}, \ \forall j, \\
&\sum_{i,j=0,1,2} \delta_{ij} = 1.
\end{split}
\end{align*}
Note that the constraints and the objective functions in the above are of the same mutual information forms as those in the BC and MAC problems in Section~\ref{sec:preliminaries}. Therefore, following the same local geometric approach,~\eqref{eq:IC_locality_assumption} can be reduced to
\begin{align} \label{eq:IC}
R_{ij}  \leq \delta_{ij} \sigma^2_{ij} , \ \mbox{for} \ i,j = 0,1,2, \ \sum_{i,j=0,1,2} \delta_{ij} \leq 1,
\end{align}
where $\sigma^2_{ij}$ indicates a channel parameter that quantifies the ability of the channel in transmitting $U_{ij}$, and can be computed in a similar manner as in Section~\ref{sec:preliminaries}:
\begin{align*}
\sigma_{ij}^2 = \left\{
                 \begin{array}{ll}
                   \sigma_{\sf smax}^2 (B_{ij}), & \hbox{$i \neq 0, j \neq 0$;} \\
                   \max_{{\bf v}_i} \min \{ || B_{i1} {\bf v}_i ||^2, || B_{i2} {\bf v}_i ||^2 \}, & \hbox{$i \neq 0, j=0$;} \\
                   \sigma_{\sf smax}^2 ([B_{1j} \; B_{2j}]), & \hbox{$i=0, j \neq 0$,} \\
                  \max_{\bf u} \min \left \{ || [B_{11} \;B_{21}] {\bf u} ||^2, || [B_{12} \;B_{22}] {\bf u} ||^2 \right \} & \hbox{$i = 0, j=0$.}
                  \end{array}
               \right.
\end{align*}
Here, $B_{ij}$ indicates the DTM with respect to the channel matrix $W_{Y_j|X_i}$ between transmitter $i$ and receiver $j$, and $({\bf v}_1, {\bf v}_2, {\bf u})$ are unit-norm vectors, such that ${\bf v}_1$ and the first $| \cX_1 |$ entries of ${\bf u}$ are orthogonal to $\sqrt{P_{X_1}}$, and ${\bf v}_2$ and the last $| \cX_2 |$ entries of ${\bf u}$ are orthogonal to $\sqrt{P_{X_2}}$. Consequently, the LIC capacity region of the interference channel is
\begin{align} \label{eq:CRIC}
{\cal C}_{\sf IC} = \bigcup_{ \sum_{ij}  \delta_{ij}   \leq 1
  } \left \{ (R_{11}, R_{10}, \cdots, R_{22}): R_{ij}  \leq \delta_{ij} \sigma_{ij}^2 \right \}.
\end{align}

\begin{figure}[t]
\begin{center}
{\epsfig{figure=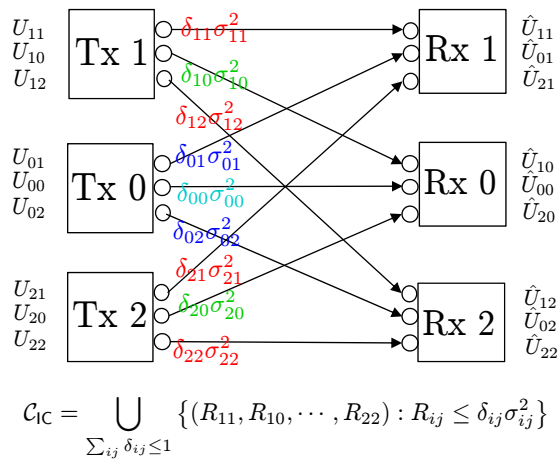, angle=0, width=0.4\textwidth}}
\end{center}
\caption{A deterministic model for interference channels. We consider the most general setting with 9 messages, denoted by $U_{ij}$'s, each indicting a message from virtual-Tx $i$ to virtual-Rx $j$. This IC can be modeled as 9 bit-pipes, each having the capacity of $\delta_{ij} \sigma_{ij}^2$, where $\delta_{ij}$ indicates the network resource assigned for transmitting $U_{ij}$.
} \label{fig:IC}
\end{figure}

From~\eqref{eq:CRIC}, we can now construct a deterministic model, applying the same idea as in the previous section. This deterministic model consists of flexible 9 bit-pipes, where the capacity of each bit-pipe is $\delta_{ij} \sigma_{ij}^2$, and can vary depending on different allocations of $\delta_{ij}$'s. An illustration of the deterministic model is shown in Fig.~\ref{fig:IC}. Note that the presented transmitters and receivers are virtual terminals, and the message $U_{ij}$ is transmitted from Tx $i$ to Rx $j$. Moreover, $\sigma_{ij}$'s should satisfy the inequalities similar to~\eqref{eq:BClambda0} and~\eqref{eq:MACsigma}:
\begin{align}
\begin{split}
\label{eq:IC_parameters_relation}
&\frac{  \sigma_{11}^2 \sigma_{12}^2  }{ \sigma_{11}^2 + \sigma_{12}^2  } \leq  \sigma_{10}^2  \leq \min \{  \sigma_{11}^2, \sigma_{12}^2 \}   \\
&\frac{  \sigma_{21}^2 \sigma_{22}^2  }{  \sigma_{21}^2 + \sigma_{22}^2  } \leq  \sigma_{20}^2 \leq \min \{  \sigma_{21}^2, \sigma_{22}^2  \}   \\
&\frac{ \sigma_{01}^2 \sigma_{02}^2  }{ \sigma_{01}^2 + \sigma_{02}^2  } \leq  \sigma_{00}^2  \leq \min \{ \sigma_{01}^2, \sigma_{02}^2  \}   \\
& \max \{  \sigma_{11}^2, \sigma_{21}^2 \} \leq  \sigma_{01}^2  \leq  \sigma_{11}^2 + \sigma_{21}^2 \\
& \max \{ \sigma_{12}^2, \sigma_{22}^2 \} \leq  \sigma_{02}^2  \leq  \sigma_{12}^2 + \sigma_{22}^2,
\end{split}
\end{align}
which can be derived similarly as in the BC and MAC cases.

\begin{example}
\label{example:IC}
Consider a quaternary-inputs binary-outputs IC where $P(y_1|x_1x_2)$ is the same as that in Example~\ref{example:MAC}, but $P(y_2|x_1x_2)$ is different as
\begin{align*}
&P(0|x_1x_2) = \left\{
                 \begin{array}{ll}
                   \frac{1}{3} ( 2- \alpha ), & \hbox{$x_1x_2=(22,23,20,32, 33,30)$;} \\
                   \alpha, & \hbox{$x_1x_2=(21, 31,01, 11)$;} \\
                   \frac{1}{3} (4 - 5 \alpha), & \hbox{$x_1x_2=(02,03,00,10)$;} \\
                   \frac{1}{3} ( -2 + 7 \alpha), & \hbox{$x_1x_2=(12,13)$,}
                 \end{array}
               \right. \\
&P(1|x_1x_2) = 1 - P(0|x_1x_2), \;\forall (x_1,x_2).
\end{align*}
To have valid probability distributions, similarly we assume that $\frac{2}{7} \leq \alpha \leq \frac{5}{7}$. Suppose that both $P_{X_1}$ and $P_{X_2}$ are fixed as $[\frac{1}{4},\frac{1}{4},\frac{1}{4}, \frac{1}{4}]^T$. The probability transition matrix $W_{ij}$ w.r.t $P_{Y_j|X_i}$ is then computed as
\begin{align*}
&W_{11} = \left[
      \begin{array}{cccc}
        \frac{1}{2} & \frac{1}{2} &1-\alpha & \alpha \\
        \frac{1}{2} & \frac{1}{2} &\alpha   & 1-\alpha \\
      \end{array}
    \right], \\
&W_{21} = \left[
      \begin{array}{cccc}
        \frac{1}{2} & \frac{1}{2} &1-\alpha & \alpha \\
        \frac{1}{2} & \frac{1}{2} &\alpha   & 1-\alpha \\
      \end{array}
    \right], \\
&W_{12} = \left[
      \begin{array}{cccc}
        1-\alpha & \alpha & \frac{1}{2} & \frac{1}{2} \\
        \alpha &  1-\alpha & \frac{1}{2}& \frac{1}{2}  \\
      \end{array}
    \right], \\
&W_{22} = \left[
      \begin{array}{cccc}
        1-\alpha & \alpha & \frac{1}{2} & \frac{1}{2} \\
        \alpha &  1-\alpha & \frac{1}{2}& \frac{1}{2}  \\
      \end{array}
    \right].
\end{align*}
This gives $B_{ij} = \frac{\sqrt{2}}{\sqrt{3}} W_{ij}$. Performing similar computations as those in Examples~\ref{example:BC} and~\ref{example:MAC}, we can get
\begin{align*}
&\sigma_{11}^2=\sigma_{12}^2 = \sigma_{21}^2= \sigma_{22}^2 = \frac{1}{2} \left(1-2\alpha \right)^2, \\
&\sigma_{10}^2 = \sigma_{20}^2 = \frac{1}{4} \left(1-2\alpha \right)^2, \\
& \sigma_{01}^2 = \sigma_{02}^2 = \left( 1-2 \alpha \right )^2, \\
& \sigma_{00}^2 = \frac{1}{2} \left(1-2\alpha \right)^2.
\end{align*}
\end{example}
This example is an extreme case where sending Rx-common messages is the hardest as possible while sending Tx-common messages is the easiest due to the maximally-achieved beamforming gain. Note that $4 \sigma_{10}^2= 2 \sigma_{11}^2 = \sigma_{01}^2$, thus implying that $(\sigma_{10}^2, \sigma_{20}^2)$ achieve the lower bounds in~\eqref{eq:IC_parameters_relation}, while $(\sigma_{01}^2, \sigma_{02}^2)$ achieve the upper bounds in~\eqref{eq:IC_parameters_relation}. $\square$

In this deterministic model, the trade-off between the 9 message rate-tuples can be characterized by solving the LP problem for $\mu$-sum-rate maximization. In particular, the LIC sum capacity can be obtained as
\begin{align*}
C_{\sf sum} &= \max_{\sum {\delta}_{ij} \leq 1} \sum \delta_{ij} \sigma_{ij}^2 = \max_{i,j} \sigma_{ij}^2 \\
&= \max \{ \sigma_{01}^2, \sigma_{02}^2 \},
\end{align*}
where the last equality is due to~\eqref{eq:IC_parameters_relation}. 
Therefore, to optimize the total throughput, we will just let either $\delta_{01}$ or $\delta_{02}$ be $1$, and deactivate other links. In other words, the optimal strategy is to transmit a common message accessible by both transmitters, maximizing the beamforming gain.

\section{Multi-hop Layered Networks} \label{sec:MLN}

Deterministic models of single-hop networks such as BCs, MACs and ICs do not well capture the trade-off between the cost of generating common messages and the benefit from sending common messages. In BCs, only the cost due to common-message generation is quantified, while in MACs, we can only investigate the benefit from common-message transmission. In ICs, an obvious solution to sum-rate maximization is to maximize the coherent combining gain which comes from common-message transmission.

On the other hand, in multi-hop layered networks, this tension can be well taken into consideration. Notice that a common message accessible by multiple transmitters in one layer must be generated from the previous layer. Hence, to optimize the throughput, one needs to compare the benefit from common-message transmission in one layer with the cost due to common-message generation in the preceding layer. Now one natural question that arises in this context is then: how do we plan which kinds of common messages should be generated in a given network to maximize the throughput? In this section, we will address this question.

\begin{figure*}
\begin{center}
{\epsfig{figure=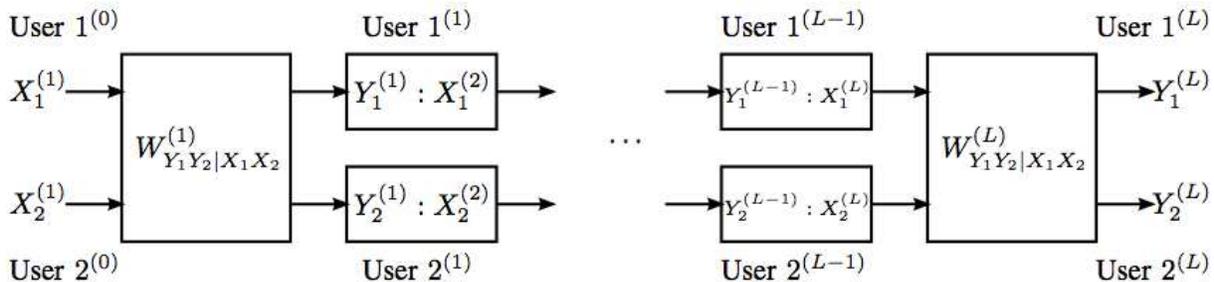, angle=0, width=0.9\textwidth}}
\end{center}
\caption{The $L$-layered network with two users in each layer. The super index ``$(\ell)$" denotes the $\ell$-th layer of the transmitters, receivers, and the users.
} \label{fig:L-network}
\end{figure*}

\begin{figure}[t]
\begin{center}
{\epsfig{figure=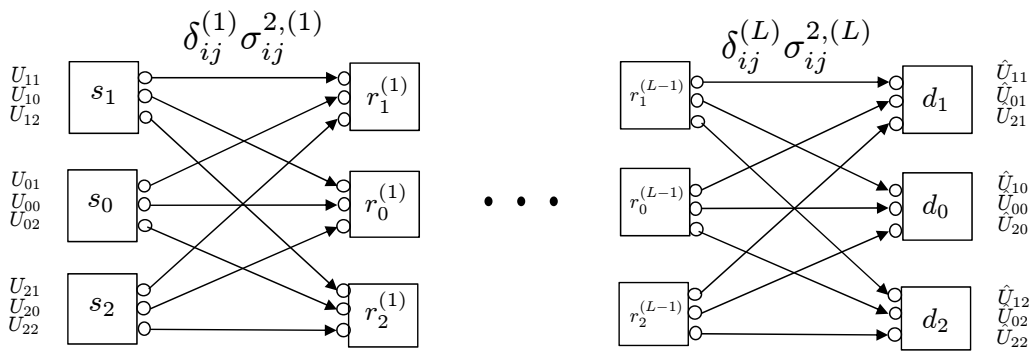, angle=0, width=0.75\textwidth}}
\end{center}
\caption{A deterministic model for multi-hop interference networks. We introduce a separation principle across layers. We abstract each layer as the bit-pipe deterministic model, and then constitute an entire network with concatenating these layers. Layer $\ell$ consists of 9 bit-pipes, each having the capacity of $\delta_{ij}^{(\ell)} \sigma_{ij}^{2,(\ell)}$, $i,j \in [0:2]$ and $\ell \in [1:L]$. Here $\sigma_{ij}^{2,(\ell)}$ represent the key parameters that characterize layer $\ell$'s channel.
} \label{fig:IC_Networks}
\end{figure}

For illustrative purpose,  we consider a general layered network with only two users in each layer, although our approach can be readily extended to more general cases at the expense of heavy notations. For the two-user $L$-layered network, the $\ell$-th layer is an interference channel with input symbols $\cX^{(\ell)}_1$, $\cX^{(\ell)}_2$, and output symbols $\cY^{(\ell)}_1$, $\cY^{(\ell)}_2$, and the channel matrix $W^{(\ell)}_{Y_1 Y_2|X_1 X_2}$. See Fig.~\ref{fig:L-network}.


For simplification, we assume a decode-and-forward operation~\cite{Cover:it79} at each layer: part of messages are decoded at each layer and then these are forwarded to the next layer. With the decode-and-forward scheme, one can abstract each layer as a deterministic model like the one for an IC, and a concatenation of these layers will construct a  deterministic model of the multi-hop layered network. See Fig.~\ref{fig:IC_Networks}. Here, we denote by $s_i$ the virtual Tx $i$ in the first layer, and by $d_i$ the virtual Rx $i$ in the last layer. Denote by $r^{(\ell)}_i$ a node that can act as the virtual Tx $i$ and Rx $i$ in the $\ell$-th intermediate layer. In addition, the channel of layer $\ell$ consists of $9$ bit-pipes, each having the capacity of $\delta_{ij}^{(\ell)} \sigma_{ij}^{2,(\ell)}$, for $i,j = 0,1,2$, and $\ell \in [1:L]$, and the corresponding constraint for $\delta_{ij}$'s is:
\begin{align}
\label{eq:IC_Network_Locality}
\frac{1}{L} \sum_{\ell=1}^{L} \sum_{i=0}^{2} \sum_{j=0}^{2} \delta_{ij}^{(\ell)} \leq 1.
\end{align}
Here the constraint is normalized by the number of layers.

For simplicity, in this paper, we do not allow any mixing between distinct messages (network coding~\cite{ahlswede:it}), focusing on the routing capacity. Then, for each set of $\delta_{ij}^{(\ell)}$ that satisfies~\eqref{eq:IC_Network_Locality}, one can obtain a layered network with fixed capacity $\delta_{ij}^{(\ell)} \sigma_{ij}^{2,(\ell)}$ for each link $(i,j)$ in the $\ell$-th layer. This reduces to the traditional routing problem. Hence, we can characterize the LIC capacity region of the $9$ rate tuples by investigating achievable rate regions over all possible sets of $\delta_{ij}^{(\ell)}$ subject to~\eqref{eq:IC_Network_Locality}.

\begin{theorem}
\label{Theorem:IC_Networks}
Consider a two-source two-destination multi-hop layered network illustrated in Fig.~\ref{fig:IC_Networks}. Assume that 9 messages $U_{ij}$'s are mutually independent. Under the assumption of~\eqref{eq:IC_Network_Locality}, the LIC capacity region is
\begin{align*}
{\cal C}_{\sf LN} = \bigcup_{ \sum  \delta_{ij}   \leq L
  } \left \{ (R_{11}, R_{10}, \cdots, R_{22}): R_{ij}  \leq \delta_{ij} \sigma_{ij}^2 \right \},
\end{align*}
where
\begin{align} \label{eq:IC_Networks_}
\sigma_{ij} = \frac{1}{L} \max_{q \in [1:3^{L-1}] }M ( {\cal P}_{ij}^{(q)} ).
\end{align}
Here, ${\cal P}_{ij}^{(q)}$ denotes a set of the link capacities along the $q$-th path from virtual source $i$ to virtual destination $j$, and $M({\cal P}_{ij}^{(q)})$ denotes the harmonic mean of the elements in the set ${\cal P}_{ij}^{(q)}$.
\end{theorem}
\begin{proof}
Unlike single-hop networks, in multi-hop networks, each link can be used for multiple purposes, i.e., $\delta_{ij}^{(\ell)}$ can be the sum of the network resources consumed for the multiple-message transmission. For conceptual simplicity, we introduce message-oriented notations $\delta_{ij}$'s, each indicating the sum of the $\delta_{ij}^{(\ell)}$'s which contribute to delivering the message $U_{ij}$. The constraint of $\sum {\delta_{ij}^{(\ell)}} \leq L$ then leads to $\sum \delta_{ij} \leq L$. Here the key observation is that the tradeoff between the 9-message rates is decided only by the constraint of $\sum \delta_{ij} \leq L$, i.e., given a fixed allocation of $\delta_{ij}$'s, the 9 sub-problems are independent with each other.

Now let us fix $\delta_{ij}$'s subject to the constraint, and consider the message $U_{ij}$. Since there are $3^{L-1}$ possible paths for transmission of this message, the problem is reduced to finding the most efficient path that maximizes $R_{ij}$, as well as finding a corresponding resource allocation for the links along the path. We illustrate the idea of solving this problem through an example in Fig.~\ref{fig:Multihop_Algorithm}. Consider the delivery of $U_{10}$. In the case of $L=2$, we have three possible paths $({\cal P}_{10}^{(1)}, {\cal P}_{10}^{(2)},{\cal P}_{10}^{(3)})$, identified by blue, red and green paths. The key point here is that the maximum rate for each path is simply a harmonic mean of the link capacities associated with the path, normalized by the number of layers. To see this, consider the top blue path ${\cal P}_{10}^{(1)}$ consisting of two links with capacities of $\sigma_{11}^{2,(1)}$ and $\sigma_{10}^{2,(2)}$, i.e., ${\cal P}_{10}^{(1)}  = \{\sigma_{11}^{2,(1)},\sigma_{10}^{2,(2)} \}$. Suppose that $\delta_{ij}$ is allocated such that the $\lambda$ fraction is assigned to the first link and the remaining $(1-\lambda)$ fraction is assigned to the second link. The rate is then computed as $\min \{ \lambda \sigma_{11}^{2,(1)}, (1-\lambda) \sigma_{10}^{2,(2)} \}$. Note that this can be maximized as $\frac{\sigma_{11}^{2,(1)} \sigma_{10}^{2,(2)} }{ \sigma_{11}^{2,(1)} + \sigma_{10}^{2,(2)} } = \frac{1}{2} M (\sigma_{11}^{2,(1)}, \sigma_{10}^{2,(2)} )$. Therefore, the maximum rate is
\begin{align*}
\sigma_{10}^{2} = \frac{1}{2} \max \left \{ M (\sigma_{11}^{2,(1)}, \sigma_{10}^{2,(2)}),  M (\sigma_{10}^{2,(1)}, \sigma_{00}^{2,(2)}), M (\sigma_{12}^{2,(1)}, \sigma_{20}^{2,(2)}) \right \}.
\end{align*}
We can easily  show that for an arbitrary $L$-layer case, the maximum rate for each path is the normalized harmonic mean. This completes the proof.
\end{proof}

\begin{figure}[t]
\begin{center}
{\epsfig{figure=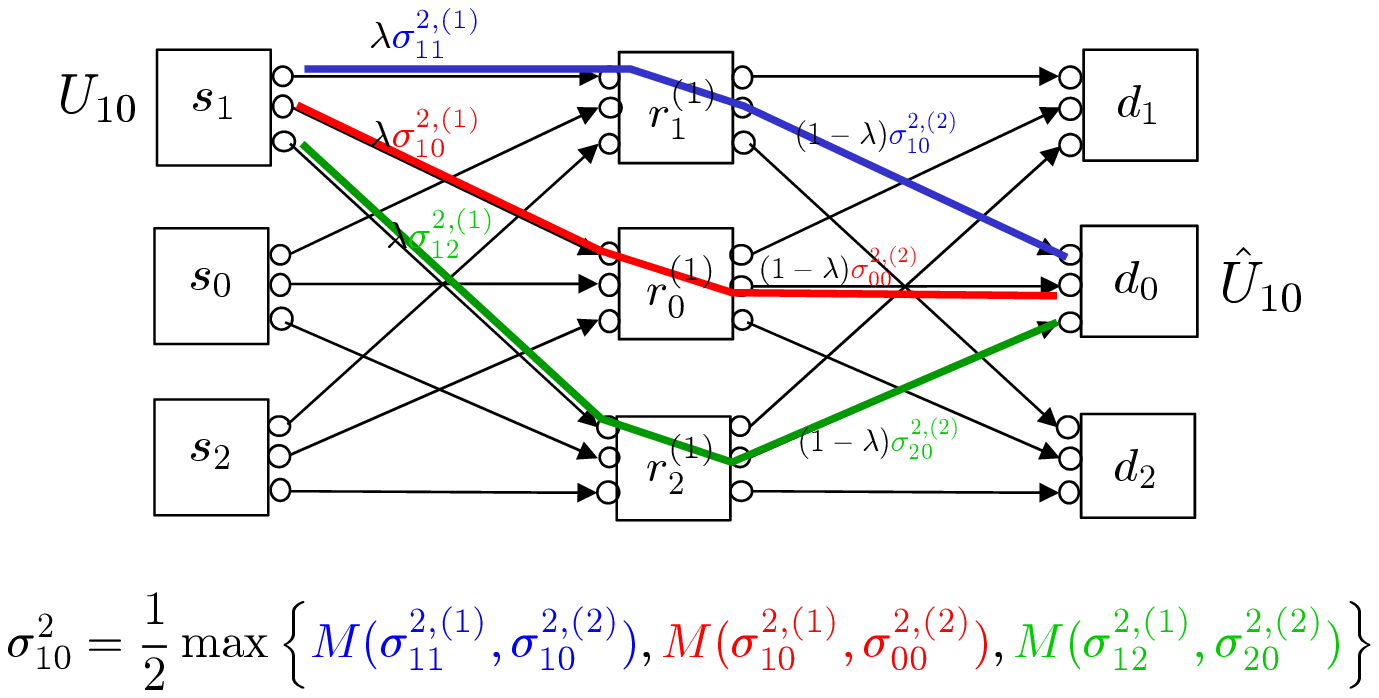, angle=0, width=0.6\textwidth}}
\end{center}
\caption{The maximum rate for $U_{10}$ when $L=2$. In this example, we have three possible paths for sending $U_{10}$ as shown in the figure. For each path, the maximum rate is computed as a harmonic mean of the link capacities along the path, normalized by the number of layers. Therefore, $\sigma_{10}^2$ is given as above.
} \label{fig:Multihop_Algorithm}
\end{figure}

\begin{remark}[Viterbi Algorithm]
Notice that the complexity for computing the LIC capacity region grows exponentially with the number of layers: $O(3^L)$. However, the Viterbi algorithm~\cite{Viterbi} allows us to reduce the complexity significantly. Note that~\eqref{eq:IC_Networks_} is equivalent to finding the path such that the inverse sum of $\sigma_{i_k i_{k+1}}^{2,(k)}$ is minimized. Taking $1/\sigma_{i_k i_{k+1}}^{2,(k)}$ as a cost, we can now apply the the Viterbi algorithm to find the path with minimal total cost, and hence the complexity is reduced to $O(L)$. $\square$
\end{remark}

In addition, Theorem~\ref{Theorem:IC_Networks} immediately provides the maximum throughput of this network as shown in the following Corollary.

\begin{corollary}
\label{Theorem:IC_Networks_sum}
Consider a layered network illustrated in Fig.~\ref{fig:IC_Networks}, the LIC sum capacity under the constraint~\eqref{eq:IC_Network_Locality} is
\begin{align} \label{eq:thm2}
C_{\sf sum} = \max_{i_1, i_2 , \ldots , i_{L+1} \in [0:2]} M(\sigma_{i_1i_2}^{2,(1)}, \sigma_{i_2i_3}^{2,(2)}, \ldots , \sigma_{i_Li_{L+1}}^{2,(L)}),
\end{align}
where $M(\cdot)$ denotes the harmonic mean.
\end{corollary}
\begin{remark}
Again one can find the optimal path via the Viterbi algorithm with complexity $O(L)$. $\square$
\end{remark}

\subsection{Multi-hop Networks with Identical Layers} \label{sec:multi-layer}

While Theorem~\ref{Theorem:IC_Networks} offers a way to find the optimal strategy for general layered networks, it is sometimes more useful to understand the ``patterns'' or structures of the optimal communication schemes for large-scaled networks. For instance, suppose that channel parameters are available only locally. Then the communication patterns can serve to design local communication strategies for users. In this section, we explore the communication patterns for a certain network: the $L$-layered network with identical channel parameters for each layer and $L \rightarrow \infty$. Specifically, for all layers $\ell$, the channel parameters are identical and denoted as $\sigma_{ij}^{2,(\ell)} = \sigma^2_{ij}$. The following theorem identifies the fundamental communication modes of the optimal strategies.

\begin{theorem}[Identical layers]
\label{Theorem:IC_SymmetricNetworks}
Consider a layered network illustrated in Fig.~\ref{fig:IC_Networks}, where $\sigma_{ij}^{2,(\ell)} = \sigma^2_{ij}, \forall \ell$, and $L \rightarrow \infty$. Then, the LIC sum capacity is
\begin{align} \label{eq:thm1}
\begin{split}
C_{\sf sum} &= \max \left \{ \sigma_{11}^2, \sigma_{00}^2, \sigma_{22}^2, M ( \sigma_{10}^2, \sigma_{01}^2), M ( \sigma_{20}^2, \sigma_{02}^2 ), M (\sigma_{12}^2, \sigma_{21}^2 ) , M ( \sigma_{10}^2, \sigma_{02}^2, \sigma_{21}^2 ), M ( \sigma_{20}^2, \sigma_{01}^2, \sigma_{12}^2 ) \right \},
\end{split}
\end{align}
where $M(\cdot)$ denotes the harmonic mean.
\end{theorem}

\begin{proof}
Let us first prove the converse part. First observe that we use the routing-only scheme to pass information through the network. Thus, for any optimal communication scheme, we have the inflow equal to outflow for every node in the intermediate layers, i.e., for all $k$ and $\ell$,
\begin{align} \label{eq:Flow_Balance}
\sum_{i=0}^{2} \delta_{ik}^{(\ell-1)} \sigma_{ik}^2 = \sum_{j=0}^{2} \delta_{kj}^{(\ell)} \sigma_{kj}^2.
\end{align}
Moreover, for all $\ell$, the total throughput of the network is $\sum_{k ,j = 0}^2 \delta_{kj}^{(\ell)} \sigma_{kj}^2$. Now, for a network with $L$ layers, let us define a tuple of $\delta_{ij}^{(\ell)}$ as a \emph{$\gamma$-scheme}, if
\begin{align*}
\sum_{k = 0}^2 \left| \sum_{j=0}^{2} \delta_{kj}^{(1)} \sigma_{kj}^2 - \sum_{i=0}^{2} \delta_{ik}^{(L)} \sigma_{ik}^2 \right| = \gamma.
\end{align*}
Here we define $C_{{\sf sum},\gamma}^{(L)}$ as the optimal achievable throughput among all $\gamma$-schemes. Since our goal is to optimize the network throughput, it suffices to only consider $\gamma$-schemes that satisfy~\eqref{eq:Flow_Balance}. Now, we  want to show that if a $\gamma$-scheme satisfies~\eqref{eq:Flow_Balance}, then $\gamma$ is upper bounded by $2 \max_{i,j} \sigma_{ij}^2$, and not increasing with $L$. To see this, note that
\begin{align*}
\gamma \leq \sum_{k = 0}^2  \sum_{j=0}^{2} \delta_{kj}^{(1)} \sigma_{kj}^2 + \sum_{k = 0}^2  \sum_{i=0}^{2} \delta_{ik}^{(L)} \sigma_{ik}^2
=  2 C_{{\sf sum},\gamma}^{(L)}  \leq 2 \max_{i,j} \sigma_{ij}^2,
\end{align*}
where the first inequality is the triangle inequality, and the second equality comes from the fact that the inflow is equal to the outflow for the schemes achieving the optimal network throughput~\eqref{eq:Flow_Balance}; hence, $\sum_{k = 0}^2  \sum_{j=0}^{2} \delta_{kj}^{(1)} \sigma_{kj}^2 = \sum_{k = 0}^2  \sum_{i=0}^{2} \delta_{ik}^{(L)} \sigma_{ik}^2 = C_{{\sf sum},\gamma}^{(L)}$. Finally, the last inequality is a trivial upper bound for the network throughput.

Now, the key technique to find the optimal throughput of the $L$-layered network is to reduce the $L$-layered optimization problem to a single-layered one. This is illustrated as follows: for any $\gamma$-scheme $\delta_{ij}^{(\ell)}$ of a network with $L$ layers that achieves $C_{{\sf sum},\gamma}^{(L)}$ and satisfies~\eqref{eq:Flow_Balance}, we consider the tuple $\tilde{\delta}_{ij}$ for $i,j = 0,1,2$, where
\begin{align*}
\tilde{\delta}_{ij} = \frac{1}{L} \sum_{\ell = 1}^{L} \delta_{ij}^{(\ell)}.
\end{align*}
Then, we have
\begin{align*}
&\sum_{k = 0}^2 \left| \sum_{j=0}^{2} \tilde{\delta}_{kj} \sigma_{kj}^2 - \sum_{i=0}^{2} \tilde{\delta}_{ik} \sigma_{ik}^2 \right| \\
=& \frac{1}{L} \sum_{k = 0}^2 \left| \sum_{j=0}^{2} \sum_{\ell = 1}^{L} \delta_{kj}^{(\ell)} \sigma_{kj}^2 - \sum_{i=0}^{2} \sum_{\ell = 1}^{L} \delta_{ik}^{(\ell)} \sigma_{ik}^2 \right| \\
=& \frac{1}{L} \sum_{k = 0}^2 \left| \sum_{\ell = 1}^{L} \sum_{j=0}^{2}  \delta_{kj}^{(\ell)} \sigma_{kj}^2 - \sum_{\ell = 2}^{L+1} \sum_{i=0}^{2}  \delta_{ik}^{(\ell-1)} \sigma_{ik}^2 \right| \\
=& \frac{1}{L} \sum_{k = 0}^2 \left| \sum_{j=0}^{2} \delta_{kj}^{(1)} \sigma_{kj}^2 - \sum_{i=0}^{2} \delta_{ik}^{(L)} \sigma_{ik}^2 \right| = \frac{\gamma}{L}.
\end{align*}
Therefore, $\tilde{\delta}_{ij}$ is a $(\gamma/L)$-scheme for a new network with only one layer, and this single layer is identical to each of the $L$ layers of the original $L$-layered network. Moreover, from~\eqref{eq:Flow_Balance}, for the $\gamma$-scheme $\delta_{ij}^{(\ell)}$ of the original $L$-layered network, the inflow and outflow of all layers are the same. So, the total throughput of the $(\gamma/L)$-scheme $\tilde{\delta}_{ij}$ of the new single-layered network is
\begin{align*}
\sum_{k = 0}^2 \sum_{j=0}^{2} \tilde{\delta}_{kj} \sigma_{kj}^2 = \frac{1}{L} \sum_{\ell = 1}^{L} \sum_{k = 0}^2 \sum_{j=0}^{2}  \delta_{kj}^{(\ell)} \sigma_{kj}^2
= \sum_{k = 0}^2 \sum_{j=0}^{2}  \delta_{kj}^{(1)} \sigma_{kj}^2 = C_{{\sf sum},\gamma}^{(L)}.
\end{align*}
This implies that ${C}_{{\sf sum},\gamma}^{(L)} \leq {C}_{{\sf sum}, \frac{\gamma}{L}}^{(1)}$. Thus, ${C}_{{\sf sum}, \frac{\gamma}{L}}^{(1)}$ is an upper bound for ${C}_{{\sf sum},\gamma}^{(L)}$, and we only need to show that $\lim_{L \rightarrow \infty}{C}_{{\sf sum}, \frac{\gamma}{L}}^{(1)}$ converges to the right hand side of~\eqref{eq:thm1}. To this end, let us first show that ${C}_{{\sf sum}, \frac{\gamma}{L}}^{(1)} $ is continuous at $\frac{\gamma}{L} = 0$.
\begin{lemma} \label{lemma1}
$\lim_{\varepsilon \rightarrow 0^{+}}{C}_{{\sf sum}, \varepsilon}^{(1)} = {C}_{{\sf sum}, 0}^{(1)}$.
\end{lemma}
\begin{proof}
See Appendix~\ref{app}.
\end{proof}

Now, note that $\gamma$ is bounded by the constant $2 \max_{i,j} \sigma_{ij}^2$, independent of $L$, so $\frac{\gamma}{L} \rightarrow 0$ in the limit of $L$. Hence, we have
\begin{align} \label{eq:(23)}
{C}_{{\sf sum}} \leq \lim_{L \rightarrow \infty} {C}_{{\sf sum},  \frac{\gamma}{L}}^{(1)} = {C}_{{\sf sum}, 0}^{(1)},
\end{align}
where the limit exists due to the continuity at $\frac{\gamma}{L} = 0$. Therefore, an upper bound of ${C}_{{\sf sum}}$ can be found by the following optimization problem:
\begin{align*}
{C}_{\sf sum} & \leq \max_{\delta_{ij}} \sum_{i,j} \delta_{ij} \sigma_{ij}^2: \\
&\textrm{ s.t. } \sum_{i,j} \delta_{ij} \leq 1,  \;\; \delta_{ij} \geq 0 \; \forall i,j \\
&\quad \;\;\; \sum_{i=0}^{2} \delta_{ik} \sigma_{ik}^2 = \sum_{j=0}^{2} \delta_{kj} \sigma_{kj}^2, \; k \in [0:2].
\end{align*}
Note that the objective indicates the total amount of information that flows into the destinations. The three equality constraints in the above can be equivalently written as two equality constraints:
\begin{align}
\begin{split}
\label{eq:SymmetricNetwork_detail1}
&\delta_{01} = \left( \frac{ \sigma_{10}^2 }{\sigma_{01}^2} \right) \delta_{10} +
\left( \frac{\sigma_{20}^2}{\sigma_{01}^2} \right) \delta_{20} - \left( \frac{\sigma_{02}^2}{\sigma_{01}^2} \right) \delta_{02} \\
&\delta_{12} = \left( \frac{\sigma_{20}^2}{\sigma_{12}^2} \right) \delta_{20} +
\left( \frac{\sigma_{21}^2}{\sigma_{12}^2} \right) \delta_{21} - \left( \frac{\sigma_{02}^2}{\sigma_{12}^2} \right) \delta_{02}.
\end{split}
\end{align}
Note that all of the $\delta_{ij}$'s are non-negative, we take a careful look at the minus terms associated with $\delta_{02}$. This leads us to consider two cases: $(1)$ $\delta_{02}=0$; $(2)$ $\delta_{02} \neq 0$.

The first is an easy case. For $\delta_{02}=0$, the problem can be simplified into:
\begin{align*}
&\max_{\delta_{ij}} \sum_{i=0}^{2} \delta_{ii} \sigma_{ii}^2 + ( 2 \delta_{10} \sigma_{10}^2
+ 3 \delta_{20} \sigma_{20}^2 + 2 \delta_{21} \sigma_{21}^2): \\
&\textrm{ s.t. } \sum_{i=0}^{2} \delta_{ii}  + \left( 1 + \frac{ \sigma_{10}^2 }{\sigma_{01}^2} \right) \delta_{10} + \left( 1 + \frac{ \sigma_{21}^2 }{\sigma_{12}^2} \right) \delta_{21} \\
 &\qquad  + \left( 1 + \frac{ \sigma_{20}^2 }{\sigma_{01}^2} + \frac{ \sigma_{20}^2 }{\sigma_{12}^2} \right) \delta_{20} \leq 1, \;\; \delta_{ij} \geq 0, \;\forall i,j.
\end{align*}
This LP problem is straightforward. Due to the linearity, the optimal solution must be setting only one $\delta_{ij}$ as a non-trivial maximum value while making the other allocations zeros. Hence, we obtain:
\begin{align}
\begin{split}
\label{eq:IN_identical_bound}
C_{\sf sum}  \leq \max \left \{ \sigma_{11}^2, \sigma_{00}^2, \sigma_{22}^2, M (\sigma_{10}^2, \sigma_{01}^2), M (\sigma_{12}^2, \sigma_{21}^2 ), M ( \sigma_{20}^2, \sigma_{01}^2, \sigma_{12}^2 ) \right \}.
\end{split}
\end{align}
Here, the fourth term $M (\sigma_{10}^2, \sigma_{01}^2)$, for example, is obtained when $\delta_{10} = \frac{1}{1 + \sigma_{10}^2 / \sigma_{01}^2}$ and $\delta_{ij} = 0$ for $(i,j) \neq (1,0)$. The last term $M ( \sigma_{20}^2, \sigma_{01}^2, \sigma_{12}^2 )$ corresponds to the case when $\delta_{20} = \frac{1}{ 1 + \sigma_{20}^2 / \sigma_{01}^2 + \sigma_{20}^2 / \sigma_{12}^2}$ and $\delta_{ij} = 0$ for $(i,j) \neq (2,0)$.

We next consider the second case of $\delta_{02} \neq 0$. First note that since $\delta_{01}$ and $\delta_{12}$ are nonnegative, by~\eqref{eq:SymmetricNetwork_detail1}, we get
\begin{align*}
& \delta_{02} \leq \left( \frac{\sigma_{20}^2}{\sigma_{02}^2} \right) \delta_{20} +  \left( \frac{\sigma_{10}^2}{\sigma_{02}^2} \right) \delta_{10},\\
& \delta_{02} \leq \left( \frac{\sigma_{20}^2}{\sigma_{02}^2} \right) \delta_{20} +   \left( \frac{\sigma_{21}^2}{\sigma_{02}^2} \right) \delta_{21}.
\end{align*}
The key point here is that in general LP problems, whenever $\delta_{02} \neq 0$, the optimal solution occurs when $\delta_{02}$ is the largest as possible and the above two inequalities are balanced:
\begin{align*}
&\delta_{02} = \left( \frac{\sigma_{20}^2}{\sigma_{02}^2} \right) \delta_{20} + \left( \frac{\sigma_{10}^2}{\sigma_{02}^2} \right) \delta_{10}, \\
 \left( \frac{\sigma_{10}^2}{\sigma_{02}^2} \right) &\delta_{10} = \left( \frac{\sigma_{21}^2}{\sigma_{02}^2} \right) \delta_{21}.
\end{align*}
Therefore, for $\delta_{02} \neq 0$, the problem can be simplified into:
\begin{align*}
&\max_{\delta_{ij}} \sum_{i=0}^{2} \delta_{ii} \sigma_{ii}^2 + ( 3 \delta_{10} \sigma_{10}^2
+ 2 \delta_{20} \sigma_{20}^2): \\
&\textrm{ s.t. } \sum_{i=0}^{2} \delta_{ii}  + \left( 1 + \frac{ \sigma_{10}^2 }{\sigma_{02}^2} + \frac{ \sigma_{10}^2 }{\sigma_{21}^2}  \right) \delta_{10} + \left( 1 + \frac{ \sigma_{20}^2 }{\sigma_{02}^2} \right) \delta_{20} \\
&\qquad  \delta_{ij} \geq 0, \;\forall i,j.
\end{align*}
This LP problem is also straightforward. Using the linearity, we can get:
\begin{align}
\begin{split}
\label{eq:IN_identical_bound2}
C_{\sf sum} \leq \max \left \{ \sigma_{11}^2, \sigma_{00}^2, \sigma_{22}^2,
 M (\sigma_{20}^2, \sigma_{02}^2), M ( \sigma_{10}^2, \sigma_{02}^2, \sigma_{21}^2 ) \right \}.
\end{split}
\end{align}
By~\eqref{eq:IN_identical_bound} and~\eqref{eq:IN_identical_bound2}, we complete the converse proof.


For the achievability, note that $\sigma^2_{i i} = M(\sigma^2_{i i})$, so all 8 modes in~\eqref{eq:thm1} can be written in the form $M(\sigma^2_{i_1 i_2}, \sigma^2_{i_2 i_3}, \ldots , \sigma^2_{i_k i_1})$, for $k = 1,2,3$, and $i_1, \ldots , i_k$ are mutually different. Then, for $k = 1,2,3$, $n \in [1:k]$, and $\ell \in [1:L]$, the $M(\sigma^2_{i_1 i_2}, \sigma^2_{i_2 i_3}, \ldots , \sigma^2_{i_k i_1})$ can be achieved by setting
\begin{align} \label{eq:abc}
\delta_{i_n i_{n+1}}^{(\ell)} = \delta_{i_n i_{n+1}} = \frac{M(\sigma^2_{i_1 i_2}, \sigma^2_{i_2 i_3}, \ldots , \sigma^2_{i_k i_1}) }{k \sigma^2_{i_n i_{n+1}}},
\end{align}
and deactivating all other links by setting their $\delta_{ij}$'s to zero. Here, we assume that in \eqref{eq:abc}, when $n = k$, $\delta_{i_k i_{k+1}}$ denotes $\delta_{i_k i_{1}}$. It is easy to verify that the assignment of \eqref{eq:abc} satisfies the constraint \eqref{eq:IC_Network_Locality}, thus we prove the achievability.
\end{proof}

\begin{figure}[t]
\begin{center}
{\epsfig{figure=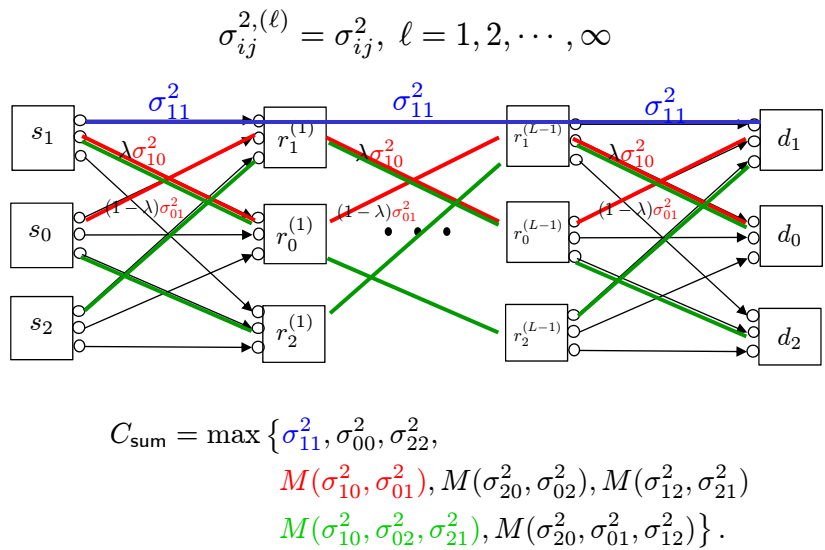, angle=0, width=0.6\textwidth}}
\end{center}
\caption{LIC sum capacity of multi-hop interference networks with identical layers.
} \label{fig:IC_SymmetricNetworks}
\end{figure}

Theorem~\ref{Theorem:IC_SymmetricNetworks} implies that the optimal communication scheme is from one of the eight communication modes in~\eqref{eq:thm1}. Fig.~\ref{fig:IC_SymmetricNetworks} illustrates the communication schemes that achieves modes $\sigma^2_{00}$, $M(\sigma^2_{12},\sigma^2_{21})$, and $M(\sigma^2_{10},\sigma^2_{02},\sigma^2_{21})$, where other modes can be achieved similarly. For example, the mode $M(\sigma^2_{10},\sigma^2_{02},\sigma^2_{21})$ is achieved by using links $1-0$, $0-2$, and $2-1$, such that
\begin{align*}
\delta_{10} \sigma^2_{10} = \delta_{02} \sigma^2_{02} = \delta_{21} \sigma^2_{21} = \frac{M(\sigma^2_{10},\sigma^2_{02},\sigma^2_{21})}{3},
\end{align*}
and other $\delta_{ij} = 0$. Then, the information flow for each layer and the sum rate are all $M(\sigma^2_{10},\sigma^2_{02},\sigma^2_{21})$.

\begin{figure}
\centering
\subfigure[]{
{\epsfig{figure=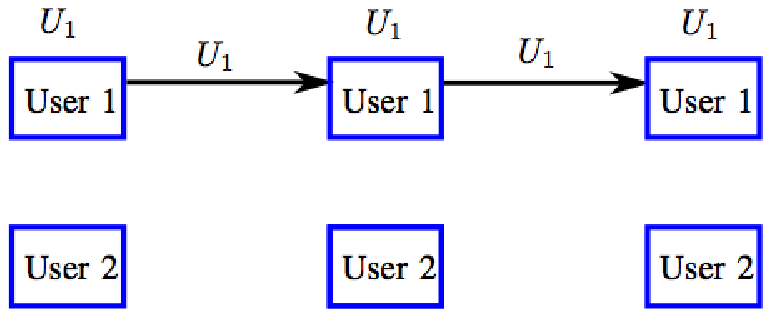, angle=0, width=0.4\textwidth}}
\label{fig:example1_4}
}
\subfigure[]{
{\epsfig{figure=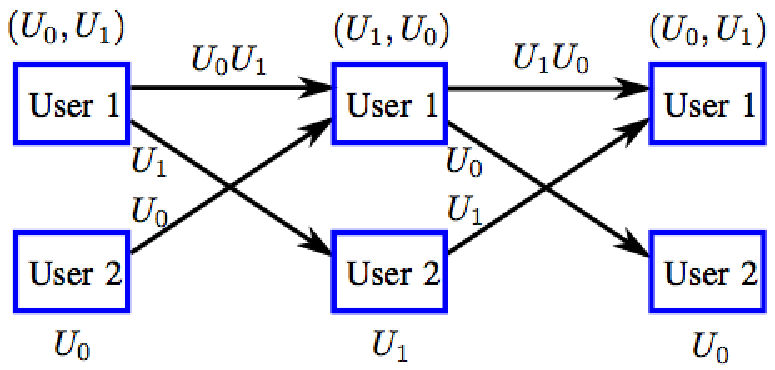, angle=0, width=0.4\textwidth}}
\label{fig:example1_3}
}
\subfigure[]{
{\epsfig{figure=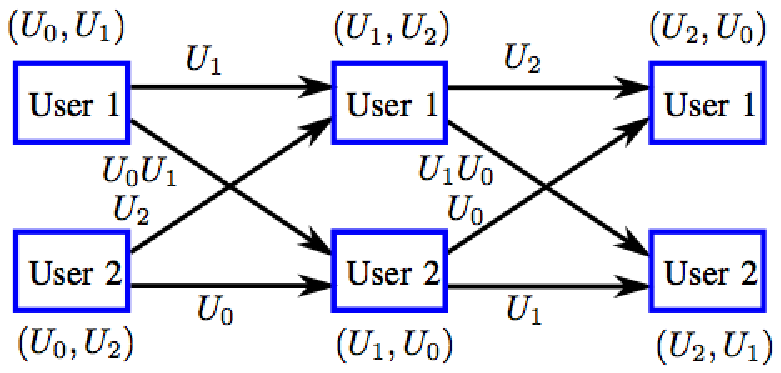, angle=0, width=0.4\textwidth}}
\label{fig:example1_2}
}
\caption{The rolling of different pieces of information between users layer by layer for the optimal communication scheme that achieves (a) $\sigma^2_{11}$ (b) $M ( \sigma^2_{10}, \sigma^2_{01})$ (c) $M ( \sigma^2_{10}, \sigma^2_{02}, \sigma^2_{21} )$.
}
\end{figure}

More interestingly, in order to achieve~\eqref{eq:thm1}, it requires the cooperation between users, and rolling the knowledges of different part of messages between users layer by layer. We demonstrate this by considering the communication scheme that achieves $M ( \sigma^2_{10}, \sigma^2_{02}, \sigma^2_{21} )$ as an example. Suppose that at the first layer, the node $s_i$ has the knowledge of message $U_i$, for $i = 0,1,2$. Since $s_0$ is the virtual node that represents the common message of both users, user $1$ knows messages $(U_0, U_1)$, and user $2$ knows $(U_0, U_2)$. Then, to achieve $M ( \sigma^2_{10}, \sigma^2_{02}, \sigma^2_{21} )$, user $1$ broadcasts its private message $U_1$ to both users in the next layer, and both users in the first layer cooperate to transmit their common message to user $2$ in the next layer as the private message. Thus, in the second layer user $1$ decodes messages $(U_1, U_2)$ and user $2$ decodes $(U_1, U_0)$. Similarly, in the third layer, user $1$ decodes $(U_2, U_0)$ and user $2$ decodes $(U_2, U_1)$, and then loop back. This effect is shown by Fig.~\ref{fig:example1_2}. Therefore, according to the values of channel parameters, Theorem~\ref{Theorem:IC_SymmetricNetworks} demonstrates the optimal communication mode, and hence indicates what kind of common messages should be generated to achieve the optimal sum rate.

\section{Feedback}
\label{sec:feedback}

We next explore the role of feedback under our local geometric approach. As in the previous section, we employ the decode-and-forward scheme for both forward and feedback transmissions, under which decoded messages at each node (instead of analog received signals) are fed back to the nodes in preceding layers. In this model, one can view the feedback as bit-pipe links added on top of a deterministic channel. 
With this assumption on the feedback, we can see that in the deterministic model of the BC, as received signals are functions of transmitted signals, so is feedback. Therefore, feedback does not increase the LIC capacity region. The deterministic MAC can be interpreted as three parallel point-to-point channels, where feedback is shown to be useless in increasing the traditional capacity~\cite{shannon:it}. Hence, the LIC capacity region does not increase with feedback either. In contrast, we will show that feedback can indeed increase the LIC capacity region for a variety of scenarios in multi-hop layered networks. Let us start with interference channels.

\subsection{Interference Channels}

\begin{theorem}
\label{theorem:IC_Feedback}
Consider the deterministic model of interference channels illustrated in Fig.~\ref{fig:IC}. Assume that decoded messages at each receiver are fed back to all the transmitters. Let $\delta_{ij}$ be the network resource consumed for delivering the message $U_{ij}$, and assume $\sum \delta_{ij} \leq 1$. The feedback LIC capacity region is then
\begin{align*}
{\cal C}_{\sf IC}^{\sf fb} = \bigcup_{ \sum  \delta_{ij}   \leq 1
  } \left \{ (R_{11}, \cdots, R_{22}): R_{k0}  \leq \delta_{k0} \sigma_{k0}^{2, {\sf fb}}, \; k \neq 0, \right. \\
  \left. R_{ij}  \leq \delta_{ij} \sigma_{ij}^2, \; (i,j) \neq (1,0), (2,0)  \right \},
\end{align*}
where
\begin{align}
\begin{split}
\label{eq:lambda_fb}
&\sigma_{10}^{2,{\sf fb}} = \max \left \{ \sigma_{10}^2, \frac{M ( \sigma_{12}^2, \sigma_{01}^2)}{2}, \frac{M ( \sigma_{11}^2, \sigma_{02}^2)}{2} \right \}, \\
&\sigma_{20}^{2,{\sf fb}} = \max \left \{ \sigma_{20}^2, \frac{M ( \sigma_{21}^2, \sigma_{02}^2)}{2}, \frac{M ( \sigma_{22}^2, \sigma_{01}^2)}{2} \right \}.
\end{split}
\end{align}
\end{theorem}
\begin{proof}
Fix $\delta_{ij}$'s subject to the constraint. First, consider the transmission of $U_{ij}$ when $(i,j) \neq (1,0), (2,0)$. In this case, the maximum rate can be achieved by using the Tx $i$-to-Rx $j$ link. Hence, $R_{ij} \leq \delta_{ij} \sigma_{ij}^2$.

\begin{figure}[t]
\begin{center}
{\epsfig{figure=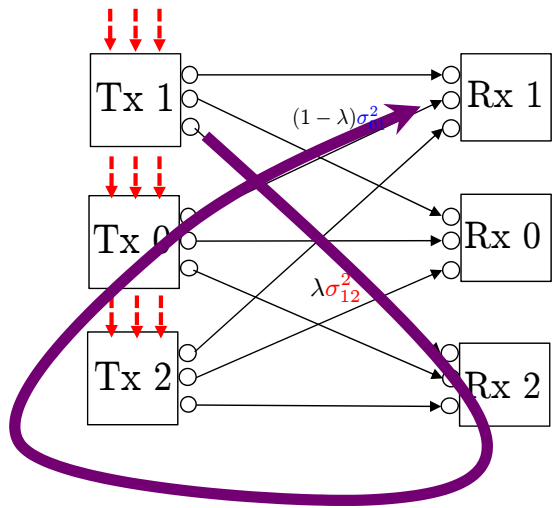, angle=0, width=0.4\textwidth}}
\end{center}
\caption{An alternative way to deliver the common message of $U_{10}$. One alternative way is to take a route: virtual-Tx 1 $\rightarrow$ virtual-Rx 2 $\stackrel{{\sf feedback}}{\longrightarrow}$ virtual-Tx 0 $\rightarrow$ virtual-Rx 1. The message is clearly a common message intended for both receivers, as it is delivered to both virtual-Rxs. We can optimize the allocation to the two links to obtain the rate of $\frac{1}{2} M (\sigma_{12}^2, \sigma_{01}^2)$.
} \label{fig:Feedback_Alternative}
\end{figure}

On the other hand, in sending $U_{10}$, we may have better alternative paths. One alternative way is to take a route as shown in Fig.~\ref{fig:Feedback_Alternative}: Tx 1 $\rightarrow$ Rx 2 $\overset{{\sf feedback}}{\longrightarrow}$ Tx 0 $\rightarrow$ virtual-Rx 1. The message is clearly a common message intended for both receivers, as it is delivered to both virtual-Rxs. Suppose that the network resource $\delta_{10}$ is allocated such that the $\lambda$ fraction is assigned to the $\sigma_{12}^2$-capacity link and the remaining $(1-\lambda)$ fraction is assigned to the $\sigma_{01}^2$-capacity link. The rate is then $\min \{ \lambda \sigma_{12}^2, (1-\lambda) \sigma_{01}^2 \}$, which can be maximized as $\frac{1}{2} M (\sigma_{12}^2, \sigma_{01}^2)$. The other alternative path is: virtual-Tx 1 $\rightarrow$ virtual-Rx 1 $\overset{{\sf feedback}}{\longrightarrow}$ virtual-Tx 0 $\rightarrow$ virtual-Rx 2.
 With this route, we can achieve $\frac{1}{2} M (\sigma_{11}^2, \sigma_{02}^2)$.
Therefore, we can obtain $\sigma_{10}^{2, {\sf fb}}$ as claimed. Similarly we can get the claimed $\sigma_{20}^{2, {\sf fb}}$.
\end{proof}
\begin{remark}[Role of feedback]
In the traditional communication setting, it is well known that feedback can increase the capacity region of MACs and degraded BCs~\cite{Ozarow:it,Ozarow:it2}, but the capacity improvement is marginal, providing at most a constant number of bits in the Gaussian channel. On the other hand, feedback can provide a more significant gain in ICs: in the Gaussian channel, it provides an arbitrarily large gain as signal-to-noise ratios of the links increase~\cite{SuhTse}. In the LIC problem setting, the impact of feedback is similar yet slightly different. The difference is that for MACs and BCs, feedback has no bearing on the LIC capacity regions. However, as can be seen from Theorem~\ref{theorem:IC_Feedback}, feedback can strictly increase the LIC capacity region in the interference channels. Also the nature of the feedback gain is similar to that in~\cite{SuhTse,Kramer:it02}: relaying gain. From Fig~\ref{fig:Feedback_Alternative}, one can see that feedback provides an alternative better path, thus making the beamforming gain effectively larger compared to the nonfeedback case. Also the feedback gain can be multiplicative, which is qualitatively similar to the gain in the two-user Gaussian interference channels~\cite{SuhTse}. Here is a concrete example in which feedback provides a multiplicative gain in the LIC capacity region. $\square$
\end{remark}
\begin{example}
Consider the same interference channel as in Example~\ref{example:IC} but which includes feedback links from all receivers to all transmitters. We obtain the same $\sigma_{ij}$'s except the following two:
\begin{align*}
\sigma_{10}^{2,{\sf fb}} = \sigma_{20}^{2,{\sf fb}} = \frac{1}{3} \left(1 - 2 \alpha \right )^2 \geq \frac{1}{4} \left(1 - 2 \alpha \right )^2  = \sigma_{10}^2 = \sigma_{20}^2.
 \end{align*}
Note that $\frac{\sigma_{10}^{2,{\sf fb}}}{\sigma_{10}^2} = \frac{4}{3}$ when $\alpha \neq \frac{1}{2}$, implying a $33\%$ gain w.r.t $R_{10}$.
$\square$
\end{example}

\begin{remark}
With Theorem~\ref{theorem:IC_Feedback}, one can simply model an interference channel with feedback as a nonfeedback interference channel, in which channel parameters $(\sigma_{10}^2,\sigma_{20}^2)$ are replaced by the $(\sigma_{10}^{2,{\sf fb}},\sigma_{20}^{2,{\sf fb}})$ in~\eqref{eq:lambda_fb}. See Fig.~\ref{fig:IC_Feedback}.
\end{remark}

\begin{figure}[t]
\begin{center}
{\epsfig{figure=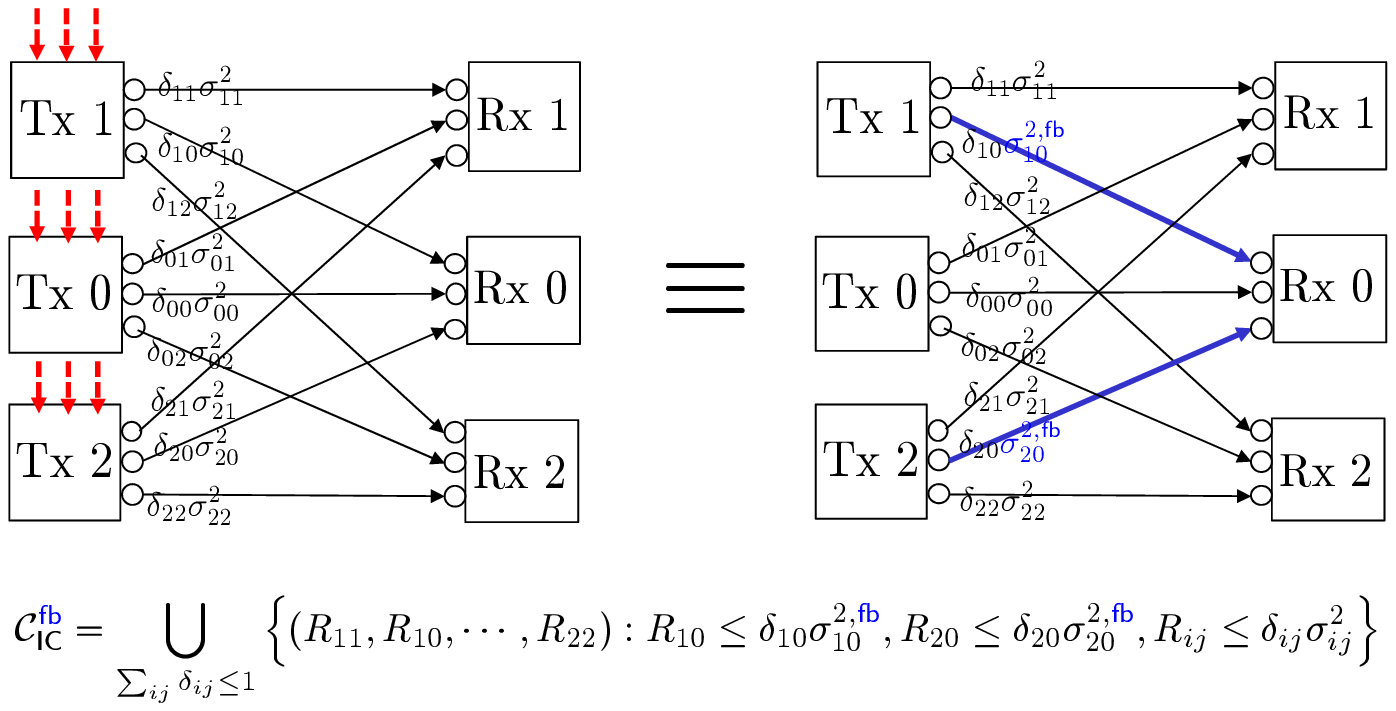, angle=0, width=0.7\textwidth}}
\end{center}
\caption{Interference channels with feedback. A feedback IC can be interpreted as a nonfeedback IC where  $(\sigma_{10}^2,\sigma_{20}^2)$ are replaced by the $(\sigma_{10}^{2,{\sf fb}},\sigma_{20}^{2,{\sf fb}})$ in~\eqref{eq:lambda_fb}.
} \label{fig:IC_Feedback}
\end{figure}

\subsection{Multi-hop Layered Networks}

For multi-hop layered networks, we investigate two feedback models: (1) full-feedback model, where the decoded messages at each node is fed back to the nodes in all the preceding layers; (2) layered-feedback model, where the feedback is available only to the nodes in the immediately preceding layer.

\begin{theorem}
\label{theorem:ICF_Networks}
Consider a multi-hop layered network illustrated in Fig.~\ref{fig:IC_Networks}. Assume that $\delta_{ij}^{(\ell)}$'s satisfy the constraint of~\eqref{eq:IC_Network_Locality}. Then, the feedback LIC capacity region of the full-feedback model is the same as that of the layered-feedback model, and is given by
\begin{align}
{\cal C}_{\sf LN}^{\sf fb} = \bigcup_{ \sum  \delta_{ij}   \leq L
  } \left \{ (R_{11}, R_{10}, \cdots, R_{22}): R_{ij}  \leq \delta_{ij} \sigma_{ij}^2 \right \},
\end{align}
where
\begin{align*}
\sigma_{ij}^2 = \frac{1}{L} \max_{1 \leq q \leq 3^{L-1}}M ( {\cal P}_{ij}^{{\sf fb},(q)} ).
\end{align*}
Here, the elements of the set ${\cal P}_{ij}^{{\sf fb},(q)}$ are with respect to a translated network where $(\sigma_{10}^{2,(\ell)},\sigma_{20}^{2,(\ell)})$ are replaced by $(\sigma_{10}^{2,(\ell),{\sf fb}},\sigma_{20}^{2,(\ell),{\sf fb}})$ for each layer $\ell \in [1:L]$:
\begin{align}
\begin{split}
\label{eq:lambda_fb_layer}
&\sigma_{10}^{2,(\ell),\sf fb} = \max \left \{ \sigma_{10}^{2,(\ell)}, \frac{M ( \sigma_{12}^{2,(\ell)}, \sigma_{01}^{2,(\ell)})}{2}, \frac{M ( \sigma_{11}^{2,(\ell)}, \sigma_{02}^{2,(\ell)})}{2} \right \}, \\
&\sigma_{20}^{2,(\ell),\sf fb} = \max \left \{ \sigma_{20}^{2,(\ell)}, \frac{M ( \sigma_{21}^{2,(\ell)}, \sigma_{02}^{2,(\ell)})}{2}, \frac{M ( \sigma_{22}^{2,(\ell)}, \sigma_{01}^{2,(\ell)})}{2} \right \}.
\end{split}
\end{align}
\end{theorem}
\begin{proof}
First, let us prove the equivalence between the full-feedback and layered-feedback models. We introduce some notations. Let $X_{i}[t]$ be the transmitted signal of virtual source $s_i$ at time $t$; let $X_{i}^{(\ell)}[t]$ be the transmitted signal of node $r_{i}^{(\ell)}$ at time $t$; and let $X^{(\ell)}[t] = \left[ X_{1}^{(\ell)}[t], X_{0}^{(\ell)}[t],X_{2}^{(\ell)}[t] \right]$, where $\ell \in [1:L-1]$. Define $X^{t-1} =\{ X[j] \}_{j=1}^{t-1}$.
 Let $Y_{i}^{(\ell)}[t]$ be the received signal of node $r_{i}^{(\ell)}$ at time $t$, and let $Y^{(\ell)}[t] = \left[ Y_{1}^{(\ell)}[t], Y_{0}^{(\ell)}[t],Y_{2}^{(\ell)}[t] \right]$, where $\ell \in [1:L]$. Let $U_i = [U_{i1},U_{i0},U_{i2}]$.
We use the notation $A \overset{f}{=} B$ to indicate that $A$ is a function $B$.

\begin{figure}[t]
\begin{center}
{\epsfig{figure=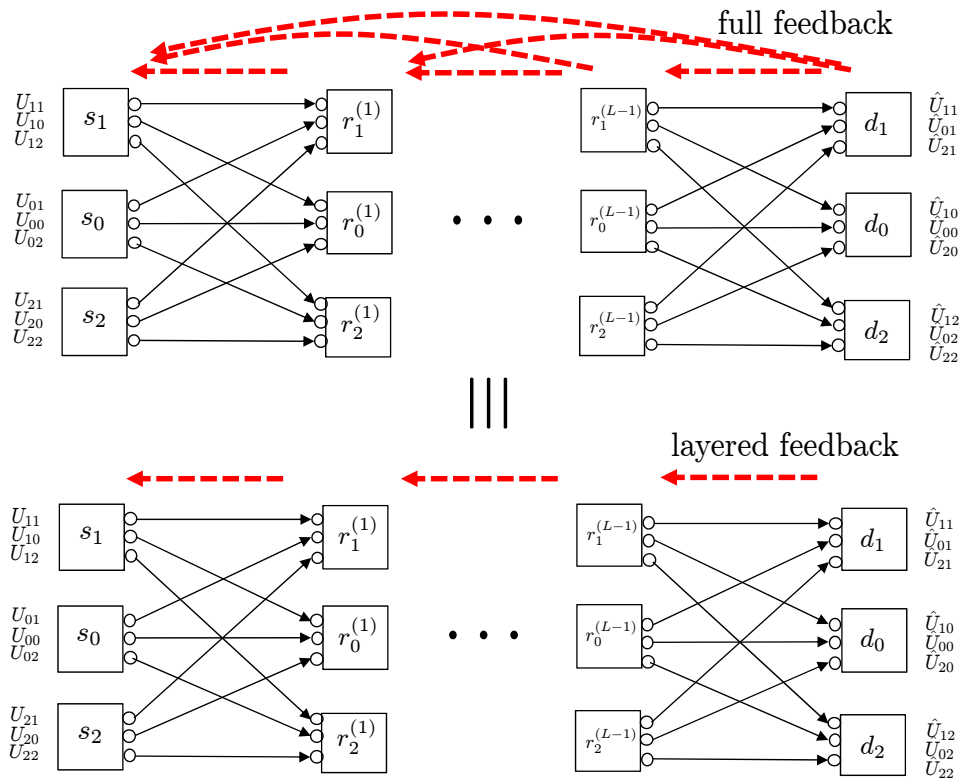, angle=0, width=0.7\textwidth}}
\end{center}
\caption{Network equivalence.  The feedback LIC capacity region of the full-feedback model is the same as that of the layered-feedback model.
} \label{fig:FB_NetworkEqui}
\end{figure}

Under the full-feedback model, we then get
\begin{align}
\begin{split}
\label{eq:full-layeredFB1}
X_{i}[t] & \overset{f}{=} (U_i, \{ Y^{(\ell),t-1} \}_{\ell = 1}^{L} ) \\
& \overset{f}{=} (U_i, Y^{(1),t-1}, X^{(1),t-1}) \\
& \overset{f}{=} (U_i, Y^{(1),t-1}, \{ Y^{(\ell),\blue{t-2}} \}_{\ell = 2}^{L}) \\
& \overset{f}{=} (U_i, Y^{(1),t-1}, X^{(1),\blue{t-2}}) \\
& \qquad \qquad \vdots \\
& \overset{f}{=} (U_i, Y^{(1),t-1}, X^{(1)}[\blue{1}]) \\
& \overset{f}{=} (U_i, Y^{(1),t-1})
\end{split}
\end{align}
where the second step follows from the fact that in deterministic layered networks, $\{ Y^{(\ell),t-1} \}_{\ell = 2}^{L}$ is a function of $X^{(1),t-1}$; the third step follows from the fact that $X^{(1),t-1} \overset{f}{=} (Y^{(1),t-2}, \{ Y^{(\ell),t-2} \}_{\ell = 2}^{L})$; and the second last step is due to iterating the previous steps $(t-3)$ times.

Using similar arguments, we can also show that for $\ell \in [1:L-1]$,
\begin{align}
\begin{split}
\label{eq:full-layeredFB2}
X_{i}^{(\ell)}[t] & \overset{f}{=} (Y_i^{(\ell),t-1}, \{ Y^{(j),t-1} \}_{j = \ell+1}^{L}) \\
& \overset{f}{=} (Y_i^{(\ell),t-1}, Y^{(\ell + 1),t-1},  X^{(\ell+1),t-1}) \\
& \overset{f}{=} (Y_i^{(\ell),t-1}, Y^{(\ell + 1),t-1}, \{ Y^{(j),\blue{t-2}} \}_{j = \ell+2}^{L}) \\
& \overset{f}{=} (Y_i^{(\ell),t-1}, Y^{(\ell + 1),t-1}, X^{(\ell+1),\blue{t-2}}) \\
& \qquad \qquad \vdots \\
& \overset{f}{=} (Y_i^{(\ell),t-1}, Y^{(\ell + 1),t-1}, X^{(\ell+1)}[\blue{1}]) \\
& \overset{f}{=} (Y_i^{(\ell),t-1}, Y^{(\ell + 1),t-1}).
\end{split}
\end{align}
The functional relationships of~\eqref{eq:full-layeredFB1} and~\eqref{eq:full-layeredFB2} imply that any rate point in the full-feedback LIC capacity region can also be achieved in the layered-feedback LIC capacity region. This proves the equivalence of the two feedback models. See Fig.~\ref{fig:FB_NetworkEqui}.

We next focus on the LIC capacity region characterization under the layered-feedback model. The key idea is to employ Theorem~\ref{theorem:IC_Feedback}, thus translating each layer with feedback into an equivalent nonfeedback layer, where $(\sigma_{10}^{2,(\ell)},\sigma_{20}^{2,(\ell)})$ are replaced by $(\sigma_{10}^{2,(\ell),{\sf fb}},\sigma_{20}^{2,(\ell),{\sf fb}})$ in~\eqref{eq:lambda_fb_layer}. We can then apply Theorem~\ref{Theorem:IC_Networks} to obtain the claimed LIC capacity region.
\end{proof}

\subsection{Multi-hop Networks with Identical Layers}

\begin{theorem}
\label{theorem:ICF_SymmetricNetworks}
Consider a multi-hop layered network in which  $\sigma_{ij}^{(\ell)} = \sigma_{ij}, \forall \ell$ and $L = \infty$. For both full-feedback and layered-feedback models, the feedback LIC sum capacity is the same as
\begin{align}
\begin{split}
C_{\sf sum}^{\sf fb} &= \max \left \{ \sigma_{11}^2, \sigma_{00}^2, \sigma_{22}^2,
M (\sigma_{10}^{2,{\sf fb}}, \sigma_{01}^2), M (\sigma_{20}^{2,{\sf fb}}, \sigma_{02}^2 ), M (\sigma_{12}^2, \sigma_{21}^2 ) ,
M ( \sigma_{10}^{2,{\sf fb}}, \sigma_{02}^2, \sigma_{21}^2 ), M ( \sigma_{20}^{2,{\sf fb}}, \sigma_{01}^2, \sigma_{12}^2 ) \right \},
\end{split}
\end{align}
where $(\sigma_{10}^{2,{\sf fb}},\sigma_{20}^{2,{\sf fb}})$ are of the same formulas as those in~\eqref{eq:lambda_fb}.
\end{theorem}
\begin{proof}
The proof is immediate from Theorems~\ref{Theorem:IC_SymmetricNetworks},~\ref{theorem:IC_Feedback}, and~\ref{theorem:ICF_Networks}. First, with Theorem~\ref{theorem:ICF_Networks}, it suffices to focus on the layered-feedback model. We then employ Theorem~\ref{theorem:IC_Feedback} to translate each layer with the layered feedback into an equivalent nonfeedback layer with the replaced parameters $(\sigma_{10}^{2,{\sf fb}},\sigma_{20}^{2,{\sf fb}})$. We can then use Theorem~\ref{Theorem:IC_SymmetricNetworks} to obtain the desired LIC sum capacity.
\end{proof}

\begin{figure}[t]
\begin{center}
{\epsfig{figure=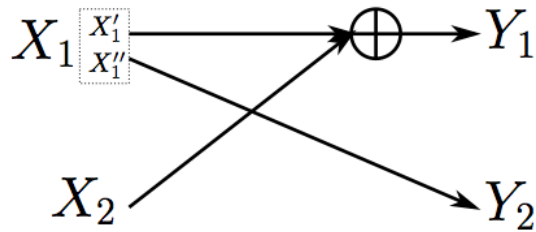, angle=0, width=0.4\textwidth}}
\end{center}
\caption{The input $X_1$ is composed of two binary inputs $X_1'$ and $X_1''$, and the input $X_2$ is binary. The output $Y_1 = X_1' \oplus X_2$, and the output $Y_2 = X_1''$.
} \label{fig:fb_ex}
\end{figure}

We see from Example 6 that the LIC sum capacity does not increase with feedback in a single-hop network. On the other hand, in multi-hop networks, we find that the LIC sum capacity can increase with feedback. Here is an example.
\begin{example}
Consider a multi-hop layered network in which each layer is the interference channel shown in Fig.~\ref{fig:fb_ex}. Tx $1$ has two binary inputs $X_1'$ and $X_1''$, and Tx $2$ has one binary input. The output $Y_1$ is equal to $X_1' \oplus X_2$ and the output $Y_2$ is equal to $X_1''$. Suppose that $P_{X_2}$ is fixed as $[0.1585 , 0.8415]$, and $P_{X_1} = P_{X_1' X_1''}$ is fixed as
\begin{align*}
&P_{X_1' X_1''} = \left\{
                 \begin{array}{ll}
                   0.095, & \hbox{$X_1' X_1'' = ( 00 , 01 )$;} \\
                   0.405, & \hbox{$X_1' X_1'' = ( 10 , 11 )$.}
                 \end{array}
               \right.
\end{align*}
Then, we have
\begin{align*}
&(\sigma_{11}^2,\sigma_{12}^2, \sigma_{10}^2)  = (0.35, 1, 0.26), \\
&(\sigma_{21}^2, \sigma_{22}^2, \sigma_{20}^2 ) = (0.25, 0, 0), \\
&(\sigma_{01}^2, \sigma_{02}^2, \sigma_{00}^2 ) = (0.6, 1, 0.375).
\end{align*}
From Theorem~\ref{Theorem:IC_SymmetricNetworks}, the nonfeedback LIC sum capacity is computed as $C_{\sf sum} = M (\sigma_{12}^2, \sigma_{21}^2 ) = 0.4$. On the other hand, $(\sigma_{10}^{2,{\sf fb}}, \sigma_{20}^{2,{\sf fb}}) = (0.375,0.2)$ and from Theorem~\ref{theorem:ICF_SymmetricNetworks}, the feedback LIC sum capacity is computed as $C_{\sf sum}^{\sf fb} =M (\sigma_{10}^{2,{\sf fb}}, \sigma_{01}^2 ) =0.4615$, thus showing a 15.4\% improvement. $\square$
\end{example}

We also find some classes of symmetric multi-hop layered networks, where feedback provides no gain in LIC sum capacity.
\begin{corollary}
Consider a two-source two-destination symmetric multi-hop layered network, where
\begin{align*}
&\lambda := \sigma_{11}^2=\sigma_{12}^2 = \sigma_{21}^2 = \sigma_{22}^2, \\
& \mu := \sigma_{10}^2 = \sigma_{20}^2, \\
& \sigma : = \sigma_{01}^2 = \sigma_{02}^2, \\
& \sigma_{00}^2.
\end{align*}
Assume that the parameters of $(\lambda,\mu,\sigma,\sigma_{00}^2)$ satisfy~\eqref{eq:IC_parameters_relation}. We then get:
\begin{align*}
C_{\sf sum} = C_{\sf sum}^{\sf fb}  = \max \{ \lambda, \sigma_{00}^2, M (\mu, \sigma), M (\mu, \lambda, \sigma) \}.
\end{align*}
\end{corollary}
\begin{proof}
Theorem~\ref{Theorem:IC_SymmetricNetworks} immediately yields $C_{\sf sum} = \max \{ \lambda, \sigma_{00}^2, M (\mu, \sigma), M (\mu, \lambda, \sigma) \}$. From Theorem~\ref{theorem:ICF_SymmetricNetworks}, we get:
\begin{align*}
C_{\sf sum}^{\sf fb}  = \max \left \{ \tilde{C}_{\sf sum}, M \left( \frac{M (\lambda, \sigma)}{2} , \sigma \right ), M \left( \frac{M (\lambda, \sigma)}{2} , \sigma, \lambda \right ) \right \}.
\end{align*}
Note that
\begin{align*}
M \left( \frac{M (\lambda, \sigma)}{2} , \sigma \right )  = \lambda \left(
\frac{2 \sigma}{2 \lambda + \sigma} \right) \leq \lambda,
\end{align*}
where the inequality comes from $\sigma \leq 2 \lambda$ due to~\eqref{eq:IC_parameters_relation}. Similarly we can show that $M \left( \frac{M (\lambda, \sigma)}{2} , \sigma, \lambda \right ) \leq \lambda$.
Therefore, $C_{\sf sum}^{\sf fb} = C_{\sf sum}$.
\end{proof}


\section{Discussions}
\label{sec:extension}

\subsection{Extension}
A generalization to arbitrary $M$-source $K$-destination networks is straightforward. In the most general setting, we have $(2^{M}-1)$ virtual sources, $(2^K-1)$ virtual destinations, and $(2^{M}-1)(2^{K}-1)$ messages. For example, in the case of $(M,K)=(3,3)$,
\begin{align*}
\textrm{virtual sources: } &s_{1}, s_{2}, s_{3}, s_{12}, s_{13}, s_{23}, s_{123}, \\
\textrm{virtual destinations: } &d_{1}, d_{2}, d_{3}, d_{12}, d_{13}, d_{23}, d_{123},
\end{align*}
where, for instance, $s_{12}$ indicates a virtual terminal that sends messages accessible by
sources 1 and 2; and $d_{12}$ denotes a virtual terminal that decodes messages intended for destinations 1 and 2. And we have $7 \times 7 = 49$ messages, denoted by $U_{ {\cal S},  {\cal D}}$, where ${\cal S}, {\cal D} \subseteq \{ 1,2,3 \} (\neq \varnothing)$, each indicating a message which is accessible by the set $\cal S$ of sources, and is intended for the set $\cal D$ of destinations. For this network, we can then obtain 49-dimensional LIC capacity regions and LIC sum capacities, as we did in Theorems~\ref{Theorem:IC_Networks} and~\ref{Theorem:IC_SymmetricNetworks}. We can also extend to networks with feedback, thus obtaining the results corresponding to Theorems~\ref{theorem:ICF_Networks} and 5.


An extension to cyclic networks is also straightforward. The key idea is to employ an unfolding technique which enables us to translate a cyclic network into an equivalent layered network. Once it is converted into a layered network, we can then apply the same techniques developed herein, thus obtaining similar results.

\subsection{Non-separation Approach \& Network Coding}

In this work, we have assumed a separation scheme between layers. Only decoded messages at each node are forwarded to next layers. We also focused on the routing capacity, not allowing for network coding. So one future research direction of interest would be developing a non-separation and/or network-coding approach to explore whether or not it provides a performance improvement over the separation approach.






\subsection{Applications of the Local Geometric Approach} \label{sec:LGA}

In this work, we took a local geometric approach based on an approximation on KL-divergence, to address a class of network information theory problems which is often quite challenging. We find this approach useful for a variety of communication scenarios and other interesting applications. As mentioned earlier in Remark~\ref{remark:ADTvsLIC}, one such communication scenario is a cognitive radio network in which the secondary users wish to exchange their information while minimizing interference to the existing communication network. Here one can model the encoding of secondary users' signals as the superposition coding to the existing primary signals. Given the constraint on the interference level, the secondary users' signals will only slightly perturb the conditional input distribution w.r.t. the primary signals from the original input distribution. Then, the decoder will detect the perturbation to decode secondary users' messages.
Therefore, our model serves to study the efficiency of exchanging information between secondary users through superposition coding, when the perturbation to the existing primary signals is restricted.

In addition to communications problems, the local geometric approach can be applied to the stock market networks. It has been shown in~\cite{HuangZheng} that the local geometric approach plays a crucial role in finding an investment strategy that maximizes an incremental growth rate in repeated investments~\cite{ErkipCover}. The local geometric approach has also been exploited to a wide range of applications in machine learning: a learning problem in graphical models~\cite{Vincent:it11}, an inference problem in hidden Markov models~\cite{Shaolun:allerton14,Shaolun:isit15}, and big networked data analytics via communication and information theory~\cite{KC15,Chuang15}.

\section{Conclusion}
\label{sec:conclusion}

In this paper, we investigate the problem of how to efficiently transmit information through discrete-memoryless networks, by perturbing the given distributions of the nodes in the networks. In particular, we apply the local approximation technique to study this problem and construct a new type of deterministic model for multi-layer networks. Then, we employ this deterministic model to investigate the optimization of the throughput of multi-layer networks. Our results illustrate the optimal communication strategy for network users to optimize the efficiency of transmitting information through large scale networks. In addition, we also consider the multi-layer networks with feedback by our deterministic model. We find that for some classes of networks, feedback can provide insights of designing efficient information flows in large communication networks.


\appendices

\section{Proof of Lemma~\ref{lemma1}} \label{app}
In this Appendix, we show that ${C}_{{\sf sum}, \varepsilon}^{(1)}$ is continuous at $\varepsilon = 0$, i.e., $\lim_{\varepsilon \rightarrow 0^{+}}{C}_{{\sf sum}, \varepsilon}^{(1)} = {C}_{{\sf sum}, 0}^{(1)}$. By the squeeze theorem, the continuity holds if the following inequalities are established: for $\varepsilon > 0$,
\begin{align} \label{eq:app1}
C_{{\sf sum}, \varepsilon}^{(1)} \geq C_{{\sf sum}, 0}^{(1)} \geq  C_{{\sf sum}, \varepsilon}^{(1)} - 4 \max_{\sigma_{ij} \neq 0} \{ \sigma_{ij}^{-2} \} \varepsilon \sum_{i,j}  \sigma_{ij}^2.
\end{align}
The upper bound of~\eqref{eq:app1} is trivial from the definition of $C_{{\sf sum}, \varepsilon}^{(1)}$. To show the lower bound of~\eqref{eq:app1}, we consider an optimal solution $\{ \delta^{*}_{ij} \}_{i,j = [0,2]}$ of ${C}_{{\sf sum}, \varepsilon}^{(1)}$, i.e., an optimal solution $\{ \delta^{*}_{ij} \}_{i,j = [0,2]}$ of the optimization problem:
\begin{align} \notag
C_{{\sf sum}, \varepsilon}^{(1)} & \leq \max_{\delta_{ij}} \sum_{i,j} \delta_{ij} \sigma_{ij}^2: \\ \notag
&\textrm{ s.t. } \sum_{i,j} \delta_{ij} \leq 1,  \;\; \delta_{ij} \geq 0 \; \forall i,j \\ \notag
&\quad \;\;\; \sum_{k=0}^2 \left| \sum_{j=0}^{2} \delta_{kj} \sigma_{kj}^2 - \sum_{i=0}^{2} \delta_{ik} \sigma_{ik}^2  \right| \leq \varepsilon. \;
\end{align}
If we can show that there exists a set of $\{ \hat{\delta}_{ij} \}_{i,j = [0,2]}$ satisfying
\begin{align} \label{the_constraint_1}
&\sum_{i,j} \hat{\delta}_{ij} \leq 1,  \;\; \hat{\delta}_{ij} \geq 0 \  \forall i,j  \\ \label{the_constraint_2}
&\sum_{k=0}^2 \left| \sum_{j=0}^{2} \hat{\delta}_{kj} \sigma_{kj}^2 - \sum_{i=0}^{2} \hat{\delta}_{ik} \sigma_{ik}^2  \right| \leq 0, \;
\end{align}
and
\begin{align}  \label{the_constraint_3}
|\delta^{*}_{ij} - \hat{\delta}_{ij}| \leq 4 \max_{\sigma_{ij} \neq 0} \{ \sigma_{ij}^{-2} \} \varepsilon, \ \forall \ i,j,
\end{align}
then from~\eqref{the_constraint_1} and~\eqref{the_constraint_2}, we have $C_{{\sf sum}, 0}^{(1)} \geq  \sum_{i,j} \hat{\delta}_{ij} \sigma_{ij}^2$. Moreover, from~\eqref{the_constraint_3}, we get
\begin{align*}
\sum_{i,j} \hat{\delta}_{ij} \sigma_{ij}^2 \geq  \sum_{i,j} \delta^*_{ij} \sigma_{ij}^2 - 4 \max_{\sigma_{ij} \neq 0} \{ \sigma_{ij}^{-2} \} \varepsilon \sum_{i,j}  \sigma_{ij}^2,
\end{align*}
which implies the lower bound of~\eqref{eq:app1}.


The idea of constructing such $\{ \hat{\delta}_{ij} \}_{i,j = [0,2]}$ is to first design each $\hat{\delta}_{ij}$ as a perturbation to $\delta^{*}_{ij}$, such that $\hat{\delta}_{ij} \geq 0$ and satisfy~\eqref{the_constraint_2} and~\eqref{the_constraint_3}. Then, the resultant $\hat{\delta}_{ij}$'s are multiplied by a normalizing factor to meet the constraint~$\sum_{i,j} \hat{\delta}_{ij} \leq 1$. To show the design of the perturbations, we define $\us_k  \triangleq\sum_{i=0}^{2} \delta^*_{ik} \sigma_{ik}^2 - \sum_{j=0}^{2} \delta^*_{kj} \sigma_{kj}^2$, where $\sum_{k = 0}^2 \us_k = 0$ from the definition. Then, since $\us_0$, $\us_1$, and $\us_2$ are symmetric w.r.t. $\sigma_{ij}$, we can without loss of generality assume $\us_0 \geq 0$, $\us_1 \geq 0$, and $\us_2 \leq 0$. In the following, we demonstrate the constructions of $\{ \hat{\delta}_{ij} \}_{i,j = [0,2]}$ for the cases of $\sigma_{20} $ and $\sigma_{21}$ being zero or nonzero:

\begin{itemize}
\item [(1)] $\sigma_{20} \neq 0 $, $\sigma_{21} \neq 0$:
\item [] In this case, we first design $\hat{\delta}_{20}$ and $\hat{\delta}_{21}$ as $\delta^*_{20} + \sigma_{20}^{-2} \us_0$ and $\delta^*_{21} + \sigma_{21}^{-2} \us_1$, and let $\hat{\delta}_{ij} = \delta^*_{ij}$ for the rest $i$ and $j$. Then, it is easy to verify that~\eqref{the_constraint_2} is satisfied. To meet the constraint $\sum_{i,j} \hat{\delta}_{ij} \leq 1$, we normalize $\hat{\delta}_{ij}$ by multiplying a factor $(1 +  \sigma_{20}^{-2} \us_0 + \sigma_{21}^{-2} \us_1)^{-1}$ to each $\hat{\delta}_{ij}$. The verification of~\eqref{the_constraint_3} for the resultant $\hat{\delta}_{ij}$ is straightforward. For example, for $\hat{\delta}_{20} = (1 +  \sigma_{20}^{-2} \us_0 + \sigma_{21}^{-2} \us_1)^{-1}(\delta^*_{20} + \sigma_{20}^{-2} \us_0)$, we have
\begin{align*}
|\delta^*_{20} - \hat{\delta}_{20}|
&\leq \left| \frac{\sigma_{20}^{-2} \us_0 + \sigma_{21}^{-2} \us_1}{1 +  \sigma_{20}^{-2} \us_0 + \sigma_{21}^{-2} \us_1} \delta^*_{20}  \right| + \left| \frac{\sigma_{20}^{-2} \us_0 }{1 +  \sigma_{20}^{-2} \us_0 + \sigma_{21}^{-2} \us_1} \right| \\
&\leq 2 \max_{\sigma_{ij} \neq 0} \{ \sigma_{ij}^{-2} \} \varepsilon  + \max_{\sigma_{ij} \neq 0} \{ \sigma_{ij}^{-2} \} \varepsilon = 3\max_{\sigma_{ij} \neq 0} \{ \sigma_{ij}^{-2} \} \varepsilon,
\end{align*}
where the first inequality is the triangle inequality, and second inequality is due to $\delta^*_{20} \leq 1$, and $\sum_{k = 0}^2 |\us_k| = \varepsilon$, which implies $|\us_k| \leq \varepsilon$, for all $k$.
\item [(2)] $\sigma_{20} \neq 0 $, $\sigma_{21} = 0$: \
\begin{itemize}
\item [(i)] $\sigma_{01} \neq 0$:
\item []
In this case, $\hat{\delta}_{01}$ and $\hat{\delta}_{20}$ are designed as $\delta^*_{01} + \sigma_{01}^{-2} \us_1$ and $\delta^*_{20} + \sigma_{20}^{-2} (\us_0 + \us_1)$. In addition, we design $\delta^*_{21}$ as $0$, and for the rest $i$ and $j$, $\hat{\delta}_{ij} = \delta^*_{ij}$. Then, it is easy to check that~\eqref{the_constraint_2} is satisfied. Moreover, we multiply each $\hat{\delta}_{ij}$ by a factor $(1+ \sigma_{01}^{-2} \us_1 + \sigma_{20}^{-2} (\us_0 + \us_1))^{-1}$ so that the constraint $\sum_{i,j} \hat{\delta}_{ij} \leq 1$ is satisfied. To verify~\eqref{the_constraint_3}, note that when $\sigma_{ij} = 0$ for some $(i,j)$, then the corresponding $\delta^*_{ij} =0$, since $\{ \delta^*_{ij} \}_{i,j = [0,2]}$ is an optimal solution. Thus, we have $\delta^*_{21} = \hat{\delta}_{21} =  0$. The verification of~\eqref{the_constraint_3} for the rest $\hat{\delta}_{ij}$'s are the same as the case (1) by noting that $|\us_0 + \us_1| \leq |\us_0| + |\us_1| \leq \varepsilon$.
\item [(ii)]  $\sigma_{01} = 0$: \
\item []
In this case, we design $\hat{\delta}_{20}$ as $\delta^*_{20} + \sigma_{20}^{-2} \us_0 + \sigma_{20}^{-2} \sigma_{10}^{2} \delta^*_{10}$, and $\delta^*_{01}$, $\delta^*_{10}$, $\delta^*_{12}$, $\delta^*_{21} $ as 0. In addition, for the rest $i$ and $j$, $\hat{\delta}_{ij} = \delta^*_{ij}$. Then, a factor $(1+ \sigma_{20}^{-2} \us_0 + \sigma_{20}^{-2} \sigma_{10}^{2} \delta^*_{10})^{-1}$ is multiplied to each $\hat{\delta}_{ij}$ for normalization. One can easily check that the resultant $\hat{\delta}_{ij}$'s satisfy both~\eqref{the_constraint_1} and~\eqref{the_constraint_2}. To verify~\eqref{the_constraint_3},  since $\sigma_{21} = \sigma_{01} = 0$, we have $\sigma_{10}^{2} \delta^*_{10} + \sigma_{12}^{2} \delta^*_{12} = \us_1 \leq \varepsilon$. Hence, $\sigma_{1k}^{2} \delta^*_{1k} \leq \varepsilon$, which implies $| \delta^*_{1k} - \hat{\delta}_{1k} |\leq \max_{\sigma_{ij} \neq 0} \{ \sigma_{ij}^{-2} \} \varepsilon$, for $k = 0,2$. Moreover, for $\hat{\delta}_{20} = (1+ \sigma_{20}^{-2} \us_0 + \sigma_{20}^{-2} \sigma_{10}^{2} \delta^*_{10})^{-1} (\delta^*_{20} + \sigma_{20}^{-2} \us_0 + \sigma_{20}^{-2} \sigma_{10}^{2} \delta^*_{10})$, we get
\begin{align*}
|\delta^*_{20} - \hat{\delta}_{20}|
&\leq \left| \frac{\sigma_{20}^{-2} \us_0 + \sigma_{20}^{-2} \sigma_{10}^{2} \delta^*_{10}}{1+ \sigma_{20}^{-2} \us_0 + \sigma_{20}^{-2} \sigma_{10}^{2} \delta^*_{10}} \delta^*_{20} \right| + |\sigma_{20}^{-2} \us_0 | + |\sigma_{20}^{-2} \sigma_{10}^{2} \delta^*_{10} | \\
&\leq \left| \sigma_{20}^{-2} \us_0 + \sigma_{20}^{-2} \sigma_{10}^{2} \delta^*_{10} \right| + |\sigma_{20}^{-2} \us_0 | + |\sigma_{20}^{-2} \sigma_{10}^{2} \delta^*_{10} | \leq 4 \max_{\sigma_{ij} \neq 0} \{ \sigma_{ij}^{-2} \} \varepsilon,
\end{align*}
where the second inequality is due to $\delta^*_{20} \leq 1$, and the third inequality is from $| \us_0 | \leq \varepsilon$ and $\sigma_{10}^{2} \delta^*_{10} \leq \varepsilon$. Finally, the verification of~\eqref{the_constraint_3} for the rest $\hat{\delta}_{ij}$'s is the same as the previous cases.
\end{itemize}
\item [(3)] $\sigma_{20} = 0 $, $\sigma_{21} \neq 0$: \
\item [] This case is symmetric to the case (2). By exchanging the role of subindexes $20 \leftrightarrow 21$, $01 \leftrightarrow 10$, and $\us_0 \leftrightarrow \us_1$, the construction is the same as the case (2),
\item [(4)] $\sigma_{20} = 0 $, $\sigma_{21} = 0$: \
\begin{itemize}
\item [(i)] $\sigma_{10} \neq 0$: \
\item []
In this case, if $\us_1 - \sigma_{02}^2 \delta^*_{02} \geq 0$, we design $\hat{\delta}_{10}$ as $\delta^*_{10} + \sigma_{10}^{-2} \us_1 - \sigma_{10}^{-2}\sigma_{02}^{2} \delta^*_{02}$; otherwise, design $\hat{\delta}_{01}$ as $ \delta^*_{01} - \sigma_{01}^{-2} \us_1 + \sigma_{01}^{-2}\sigma_{02}^{2} \delta^*_{02}$. In addition, we design $\hat{\delta}_{20}$, $\hat{\delta}_{21}$, $\hat{\delta}_{02}$, $\hat{\delta}_{12}$ to 0, and for the rest $i$ and $j$, $\hat{\delta}_{ij} = \delta^*_{ij}$'s. We multiply a factor $(1 + |\sigma_{10}^{-2} \us_1 - \sigma_{10}^{-2}\sigma_{02}^{2} \delta^*_{02}|)^{-1}$ to each $\hat{\delta}_{ij}$ for normalization. Then, one can check that~\eqref{the_constraint_1} and~\eqref{the_constraint_2} are satisfied for the resultant $\hat{\delta}_{ij}$. Note that since $\sigma_{20} = \sigma_{21} = 0$, we get $\sigma_{02}^{2} \delta^*_{02} + \sigma_{12}^{2} \delta^*_{12} = - \us_2 \leq \varepsilon$, which implies $\sigma_{k2}^{2} \delta^*_{k2} \leq \varepsilon$, for $k = 0,1$. Thus, by the same procedure as (ii) of the case (2), we can verify~\eqref{the_constraint_3}.
\item [(ii)] $\sigma_{10} = 0$: \
\item []
In this case, we simply set all the $\hat{\delta}_{ij}$'s be zero. Since $\sigma_{20} = \sigma_{21} = \sigma_{10} = 0$, we get $\sigma_{01}^{2} \delta^*_{01} + \sigma_{02}^{2} \delta^*_{02} = \us_0 \leq \varepsilon$, and $\sigma_{02}^{2} \delta^*_{02} + \sigma_{12}^{2} \delta^*_{12} = - \us_2 \leq \varepsilon$, which imply $\delta^*_{ij} \leq \max_{\sigma_{ij} \neq 0} \{ \sigma_{ij}^{-2} \} \varepsilon$ for all $i$ and $j$. Thus,~\eqref{the_constraint_1} to~\eqref{the_constraint_3} are satisfied.
\end{itemize}
\end{itemize}

\bibliographystyle{ieeetr}
\bibliography{bib_ENIT}

\end{document}